%% file: polyMethod_GD.tex
\documentclass[runningheads]{llncs}

\newif\ifarxiv
\arxivtrue  

\ifarxiv 
\else
\fi
\usepackage{graphicx}
\usepackage{paralist}
\usepackage{nicefrac}
\usepackage{hyperref}
\usepackage{xspace}

\usepackage[utf8]{inputenc}
\usepackage[T1]{fontenc}
\usepackage{amsmath}
\usepackage{amssymb}
\usepackage{cleveref}
\usepackage{wasysym}
\usepackage{pifont}
\usepackage{xcolor}
\usepackage{thmtools} 
\usepackage{thm-restate}
\usepackage{graphicx}
\usepackage{caption}
\usepackage{subcaption}
\usepackage{framed}
\usepackage{thm-restate}
\usepackage{lineno}

\spnewtheorem{clm}{Claim}{\bfseries}{\itshape}

\newenvironment{obs}{\smallbreak \noindent \textbf{Observation.}}{\smallbreak}
\newenvironment{rmk}{\smallbreak\noindent Remark.}{\smallbreak}

\newcommand{\graphClass}{accordion\xspace}
\newcommand{\GraphClass}{Accordion\xspace}

\newcommand{\R}{\ensuremath{\mathbb R}}
\newcommand{\A}{\ensuremath{\mathcal A}}
\newcommand{\As}{\ensuremath{\mathbb A}}
\newcommand{\B}{\ensuremath{\mathcal B} }
\newcommand{\D}{\ensuremath{\mathcal D}}
\newcommand{\Q}{\ensuremath{\mathbb Q}}
\newcommand{\Po}{\ensuremath{\mathcal P}\xspace}
\newcommand{\Or}{\ensuremath{\mathcal O}}
\newcommand{\f}{\ensuremath{\mathfrak f}}

\newcommand{\outdeg}{\ensuremath{\text{outdeg}}}
\newcommand{\Det}{\ensuremath{\text{Det}}}
\newcommand{\area}{\ensuremath{\textsc{area}}}
\newcommand{\diam}{dia\-mond addition\xspace}

\newcommand{\FER}{\ensuremath{\forall \exists \mathbb R}\xspace}
\newcommand{\ER}{\ensuremath{\exists \mathbb{R}}\xspace}
\newcommand{\AreaEqA}[1]{\ensuremath{\textsc{aeq}(T,\A,#1)}}
\newcommand{\AreaEq}[2]{\ensuremath{\textsc{aeq}(T,#1,#2)}}

\newcommand{\N}{\ensuremath{\mathcal N}}
\newcommand{\Nx}{\ensuremath{\mathcal{N}^x}}
\newcommand{\Ny}{\ensuremath{\mathcal N^y}}

\newcommand{\pOrder}{p-order\xspace}
\newcommand{\pred}{\text{pred}\xspace}
\newcommand{\lastfaceFun}{last face function\xspace}
\newcommand{\lkGraph}{$\ell k$-graph\xspace}
\newcommand{\crr}{crr-free\xspace}

\newcommand{\acc}[1]{\ensuremath{\mathcal K_{#1}}\xspace}
\newcommand{\dS}[2]{\ensuremath{\mathcal H_{#1,#2}}\xspace}
\def\arVP{almost realizing vertex place\-ment\xspace}

\crefname{section}{Section}{Sections}
\crefname{theorem}{Theorem}{Theorems}
\crefname{corollary}{Corollary}{Corollaries}
\crefname{lemma}{Lemma}{Lemmas}
\crefname{observation}{Observation}{Observations}
\crefname{figure}{Figure}{Figures}
\crefname{equation}{Equation}{Equations}
\crefname{clm}{Claim}{Claims}
\crefname{proposition}{Proposition}{Propositions}

\begin{document}

\title{On the Area-Universality of Triangulations}
\titlerunning{Area-universal triangulations}
\author{Linda Kleist}
\authorrunning{L. Kleist}
\institute{Technische Universit\"at Berlin, Germany, \email{kleist@math.tu-berlin.de}}
\maketitle             

\begin{abstract}
We study straight-line drawings of planar graphs with prescribed face areas. A plane graph is \emph{area-universal} if for every area assignment on the inner faces, there exists a straight-line drawing realizing the prescribed areas. 

For triangulations with a special vertex order, we present a sufficient criterion for area-universality that only requires the investigation of one area assignment.
Moreover, if the sufficient criterion applies to one plane triangulation, then all embeddings of the underlying planar graph are also area-universal. To date, it is open whether area-universality is a property of a plane or planar graph.

We use the developed machinery to present area-universal families of triangulations. Among them we characterize area-universality of accordion graphs showing that area-universal and non-area-universal graphs may be structural very similar. 
\end{abstract}

\graphicspath{{figures/}}

\input{1introduction}
\input{2prelim}
\input{2TriangCharacterization}
\input{3polynomialMethod2}

\input{5Applications}

\section{Discussion and Open Problems}\label{sec:end}
For triangulations with \pOrder{}s, we introduced a sufficient criterion to prove area-universality of all embeddings of a planar graph which relies on checking properties of one area assignments of one plane graph.  
We used the criterion to present two families of area-universal triangulations. Since area-universality is maintained by taking subgraphs, area-universal triangulations are of special interest. For instance, the area-universal double stacking graphs are used in~\cite{kleistQuad,kleistPhD} to show that all plane quadrangulations with at most 13 vertices are area-universal. 
The analysis of \graphClass graphs showns that  area-universal and non-area-universal graphs can be structural very similar.
The class of accordion graphs gives a hint why understanding area-universality seems to be a difficult problem. 
In conclusion, we pose the following open questions:
\begin{compactitem}
 \item Is area-universality a property of plane or planar graphs?
 \item What is the complexity of deciding the area-universality of triangulations? 
 \item Can area-universal graphs be characterized by \emph{local} properties?
\end{compactitem}

\paragraph{Acknowledgements}
I thank Udo Hoffmann and Sven Jäger for helpful comments. 
\newpage

%
%
%

\bibliography{MyBib_03}
\bibliographystyle{splncs04}

\ifarxiv
\newpage
\appendix
\input{8Appendix1}
\input{8Appendix2}
\input{8Appendix3}

\fi
\end{document}

%% file: 1introduction.tex
%

\section{Introduction}
By Fary's theorem \cite{fary,stein1951,wagner_bemerkungen}, every plane graph has a straight-line drawing. 
We are interested in straight-line drawings with the additional property that the face areas correspond to prescribed values. Particularly, we study \emph{area-universal} graphs for which all prescribed face areas can be realized by a straight-line drawing. Usually, in a \emph{planar} drawing, no two edges intersect except in common vertices. It is worthwhile to be slightly more generous and allow  \emph{crossing-free} drawings, i.e., drawings that can be obtained as the limit of a sequence of planar straight-line drawings.
Note that a crossing-free drawing of a triangulation is not planar (\emph{degenerate}) if and only if the area of at least one face vanishes.
Moreover, we consider two crossing-free drawings of a plane graph as \emph{equivalent} if the cyclic order of the incident edges at each vertex and the outer face coincide.

For a plane graph $G$, we denote the set of faces by $F$, and the set of inner faces by $F'$. An \emph{area assignment} is a function $\A\colon F'\to \mathbb \R_{\geq0}$. 
We say $G$ is \emph{area-universal} if for every area assignment $\A$ there exists an equivalent crossing-free drawing  where every inner face  $f\in F'$ has area $\A(f)$. We call such a drawing \emph{$\A$-realizing} and the area assignment $\A$ \emph{realizable}.

\paragraph{Related Work.} 

Biedl and Ruiz Vel{\'{a}}zquez \cite{biedl} showed that planar partial 3-trees, also known as subgraphs of \emph{stacked triangulations} or \emph{Apollonian networks}, are area-universal.
In fact, every subgraph of a plane area-universal graph is area-universal.
Ringel \cite{ringel} gave two examples of graphs that have drawings where all face areas are of equal size, namely the octahedron graph and the icosahedron graph. 
Thomassen \cite{thomassen}  proved that plane 3-regular graphs are area-universal. 
Moreover, Ringel \cite{ringel} showed that the octahedron graph is not area-universal. 
Kleist \cite{kleist1Journal} 
generalized this result by introducing a  simple counting argument which shows that no Eulerian triangulation, different from $K_3$, is area-universal. 
Moreover, it is shown in 
\cite{kleist1Journal}
that every  1-subdivision of a plane graphs is area-universal; that is, every area assignment of a plane graph has a realizing polyline drawing where each edge has at most one bend.
Evans et al.\ \cite{kleistQuad,kleistPhD} present classes of area-universal plane quadrangulations. In particular, they verify the conjecture that plane bipartite graphs are area-universal for quadrangulations with up to 13 vertices. Particular graphs have also been studied: It is known that the square grid \cite{tableCartograms} and the unique triangulation on seven vertices \cite{BachelorThesis} are area-universal. Moreover, non-area-universal triangulations on up to ten vertices have been investigated in~\cite{Henning}. 

The computational complexity of the decision problem of area-universality for a given graph was studied by Dobbins et al.\ \cite{kleist2}. The authors show that this decision problem belongs to \textsc{Universal Existential Theory of the Reals}  $(\FER)$, a natural generalization of the class \textsc{Existential Theory of the Reals} (\ER), and conjecture that this problem is also \FER-complete. They show hardness of several variants, e.g., the analogue problem of volume universality of simplicial complexes in three dimensions.

In a broader sense, drawings of planar graphs with prescribed face areas can be understood as {\em cartograms}.  
Cartograms have been intensely studied for duals of triangulations  \cite{cartogramsTs,deBerg07,BR-WADS11,Nagamochi} and in the context of 
rectangular layouts, dissections of a 
rectangle into rectangles \cite{Eppstein,felsner,Floorplans}. For a detailed survey of the cartogram literature, we refer to~\cite{NusratKobourovSurvery}.

\paragraph{Our contribution.} In this work we present three characterizations of area-uni\-ver\-sal triangulations. We use these characterizations for proving area-universality of certain triangulations. 
Specifically, we consider triangulations with a vertex order,  where (most) vertices have at least three neighbors with smaller index, called \emph{predecessors}. We call such an order a \emph{\pOrder}.
For triangulations with a \pOrder, 
the realizability of an area assignment reduces to finding a real root of a univariate polynomial. If the polynomial is surjective, we can guarantee area-universality. 
In fact, this is the only known method to prove the area-universality of a triangulation besides the simple argument for plane 3-trees relying on $K_4$.

We discover several interesting facts. First, to guarantee area-universality it is enough to investigate one area assignment. Second, if the polynomial is surjective for one plane graph, then it is for every embedding of the underlying planar graph. Consequently, the  properties of one area assignment can imply the area-universality of all embeddings of a planar graph. 
This may indicate that area-universality is a property of planar graphs.

We use the method to prove area-universality for several graph families including  accordion graphs. 
To obtain an \emph{\graphClass graph} from the plane octahedron graph,
 we introduce new vertices of degree 4 by subdividing an edge of the central triangle. 
\cref{fig:funnyClass2} presents four examples of \graphClass graphs. 
Surprisingly, the insertion of an even number of vertices  yields a non-area-universal graph while the insertion of an odd number of vertices  yields an area-universal graph.  Accordions with an even number of vertices are Eulerian and thus  not area-universal \cite{kleist1Journal}. Consequently, area-universal and non-area-universal graphs may have a very similar structure. (In \cite{kleistPhD}, we use the method to classify small triangulations with \pOrder{}s on up to ten vertices.)

\begin{figure}[hbt]
 \centering
 \includegraphics[page=1]{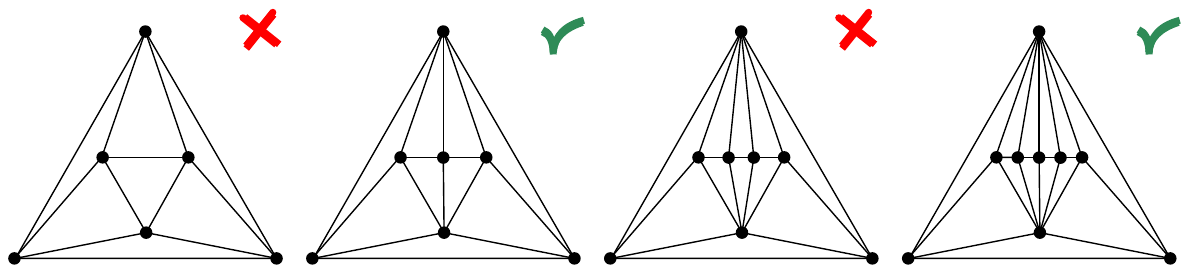}
 \caption{Examples of accordion graphs. A checkmark indicates area-universality and a cross non-area-universality.}
 \label{fig:funnyClass2}
\end{figure}

\paragraph{Organization.} 
We start by presenting three characterizations of area-universality of triangulations in \cref{sec:prelim}.
In \cref{sec:poly}, we turn our attention to triangulations with \pOrder{}s and show how the analysis of one area assignment can be sufficient to prove area-universality of all embeddings of the given triangulation.
Then, in \cref{sec:applications}, we apply the developed method to prove area-universality for certain graph families; among them we characterize the area-universality of \graphClass graphs. We end with a discussion and a list of open problems in \cref{sec:end}.
\ifarxiv
In \cref{app:2,app:3,app:applications} we present omitted proofs of \cref{sec:prelim,sec:poly,sec:applications}.
\else
For omitted proofs consider the appendices of the full version~\cite{kleist3Arx}.
\fi

%% file: 2prelim.tex
\section{Characterizations of Area-Universal Triangulations}\label{sec:prelim}
Throughout this section, let $T$ be a plane triangulation on $n$ vertices.
A straight-line drawing of $T$ can be encoded by the $2n$ vertex coordinates, and hence, by a point in the Euclidean space $\R^{2n}$.
We call such a vector of coordinates a \emph{vertex placement} and denote the set of all vertex placements encoding crossing-free drawings  by $\D(T)$; we also write $\D$ if $T$ is clear from the context. 

It is easy to see that an  $\A$-realizing drawing of a triangulation can be transformed by an affine linear map into an $\A$-realizing drawing where the outer face corresponds to any given triangle of correct total area $\Sigma \A:=\sum_{f\in F'}\A(f)$, where $F'$ denotes the set of inner faces as before.
\begin{lemma}\emph{\cite[Obs. 2]{kleist1Journal}}\label{lem:outerFace}
A plane triangulation $T$ with a realizable area assignment~$\A$, has an $\A$-realizing drawing within every given outer face of area $\Sigma\A$.
\end{lemma}
Likewise, affine linear maps can be used to scale realizing drawings by any factor. 
For any positive real number $\alpha\in\mathbb R$ and area assignment $\A$, let $\alpha \A$ denote the \emph{scaled area assignment} of $\A$ where $\alpha\A(f):=\alpha\cdot \A(f)$ for all $f\in F'$.
\begin{lemma}\label{lem:scaling}
Let $\A$ be an area assignment of a plane graph and $\alpha>0$. The scaled area assignment $\alpha\A$ is realizable if and only if $\A$ is realizable. 
\end{lemma}
For a plane graph and $c>0$, let $\As^c$ denote the set of area assignments with a total area of $c$. \cref{lem:scaling} directly implies the following property.
\begin{lemma}\label{lem:fixedTotal}
 Let $c>0$. A plane graph is area universal if all area assignments in~$\As^c$ are realizable.
\end{lemma}


 
\subsection{Closedness of Realizable Area Assignments}
In \cite[Lemma 4]{kleist1Journal}, it is shown for triangulations that $\A\in \As^c$ is realizable if and only if in every open neighborhood of $\A$ in $\As^c$ there exists a realizable area assignment.
For our purposes, we need a stronger version. 
Let $\As^{\leq c}$ denote the set of area assignments of~$T$ with a total area of at most $c$.  For a fixed face $f$ of~$T$, $\As^{\leq c}|_{f\to a}$ denotes the subset of $\As^{\leq c}$  where $f$ is assigned to a fixed $a>0$.

\begin{restatable}{proposition}{fillingStronger}\label{cor:fillinSTRONGER}
 Let $T$ be a plane triangulation and $c>0$. 
 Then $\A\in\As^c$ is realizable if and only if for some face $f$ with $\A(f)>0$ every open neighborhood of $\A$ in $\As^{\leq 2c}|_{f\to \A(f)}$ contains a realizable area assignment. 
\end{restatable}
Intuitively, \cref{cor:fillinSTRONGER} enables us not to worry about area assignments with bad but unlikely properties.  In particular, area-universality is guaranteed by the realizability of a dense subset of $\As^c$. 
Moreover, this stronger version allows to certify the realizability of an area assignment by  realizable area assignments with slightly different total areas.
The proof of \cref{cor:fillinSTRONGER} goes along the same lines as in \cite[Lemma 4]{kleist1Journal}; it is based on the fact that the set of drawings of $T$ with a fixed face $f$ and a total area of at most $2c$ is compact.

\subsection{Characterization by 4-Connected Components}
For a plane triangulation $T$, a \emph{$4$-connected component} is a maximal $4$-connected subgraph of $T$. Moreover, we call a triangle $t$ of~$T$ \emph{separating} if at least one vertex of~$T$ lies inside $t$ and at least one vertex lies outside $t$; in other words, $t$ is not a face of $T$.

\begin{restatable}{proposition}{connected}\label{lem:4connnected}
 A plane triangulation $T$ is area-universal if and only if every 4-connected component of $T$ is area-universal.
\end{restatable}
\begin{proof}[Sketch]
The proof is based on the fact that a plane graph $G$ with a separating triangle $t$ is area-universal if and only if $G_\textsc{e}$, the induced graph by~$t$ and its exterior,  and $G_\textsc{i}$,  the induced graph by $t$ and its interior,  are area-universal. In particular, \cref{lem:outerFace} allows us to combine realizing drawings of $G_\textsc{e}$ and $G_\textsc{i}$ to a drawing of $G$.
\end{proof}

\begin{rmk}
 Note that a plane 3-tree has no 4-connected component.
 (Recall that $K_4$ is 3-connected and a graph on $n>4$ vertices is 4-connected if and only if it has no separating triangle.)
 This is another way to see their area-universality.
\end{rmk}

%% file: 2TriangCharacterization.tex
\subsection{Characterization by Polynomial Equation System}
Dobbins et al.\ \cite[Proposition 1]{kleist2} show a close connection of area-universality and equation systems: For every plane graph $G$ with area assignment $\A$ there exists a polynomial equation system $\mathcal E$ such that $\A$ is realizable if and only if $\mathcal E$ has a real solution. 
Here we strengthen the statement for triangulations, namely it suffices to guarantee the face areas; these imply all further properties such as planarity and the equivalent embedding. 
To do so, we introduce some notation. 

A plane graph $G$ induces an orientation of the vertices of each face. For
a face $f$ given by the vertices $v_1 , \dots , v_k$, we say $f$ is \emph{counter clockwise (ccw)} if
the vertices $v_1 , \dots , v_k$ appear in ccw direction on a walk on the boundary of $f$; otherwise $f$ is \emph{clockwise (cw)}. Moreover, the function $\area(f,D)$ measures the area of a face $f$ in a drawing $D$.
For a ccw triangle $t$ with vertices $v_1, v_2, v_3$, we denote the coordinates of $v_i$ by $(x_i, y_i)$. Its area in~$D$ is given by the determinant
\begin{linenomath*}
 \begin{equation}\label{eq:det}
 \Det(v_1,v_2,v_3):=\det\big(c(v_1),c(v_2),c(v_3)\big)=2\cdot \area (t,D),
 \end{equation}
\end{linenomath*}
where $c(v_i):=(x_i,y_i,1)$.
Since the (complement of the) outer face $f_o$ has area $\Sigma \A$ in an $\A$-realizing drawing, we define  $\A(f_o):=\Sigma \A$.
For a set of faces $\tilde F\subset F$, we define 
the \emph{area equation system of $\tilde F$} as 
\begin{linenomath*}
\begin{equation*}
 \AreaEqA {\tilde F}:=\{\Det(v_i,v_j,v_k)=\A(f)\mid f\in \tilde F, f=:{(v_i,v_j,v_k)} \text{ ccw}\}.
\end{equation*}
\end{linenomath*}
For convenience, we  omit the factor of $2$ in each \emph{area equation}. 
Therefore, without mentioning it any further, we usually certify the realizability of $\A$ by a $\nicefrac{1}{2}\A$-real\-izing drawing. That is, if we say a triangle has area $a$, it may have area $\nicefrac{1}{2}a$. Recall that, by \cref{lem:scaling},  consistent scaling has no further implications.

\begin{restatable}{proposition}{propThree}\label{thm:areaEqs}
 Let $T$ be a triangulation, $\A$ an area assignment, and $f$ a face of $T$.  Then $\A$ is realizable if and only if  $\AreaEqA{F\setminus \{f\}}$ has a real solution.
\end{restatable}
The key idea is that a (scaled) vertex placement of an $\A$-realizing drawing is a real solution of $\AreaEqA{F\setminus \{f\}}$ and vice versa.
The main task is to guarantee crossing-freeness of the induced drawing; it follows from the following neat fact.

\begin{restatable}{lemma}{ccPlanar}\label{lem:ccPlanar}
Let $D$ be a vertex placement of a  triangulation $T$ where the orientation of each inner face in $D$ coincides with the orientation in $T$. Then $D$ represents a crossing-free straight-line drawing of $T$.
\end{restatable}
A proof of \cref{lem:ccPlanar} can be found in \cite[in the end of the proof of Lemma 4.2]{windrose}. An alternative proof relies on
the properties of the determinant, in particular, on
the fact that
for any vertex placement $D$  the area of the triangle formed by its outer vertices evaluates to 
 \begin{equation}\label{eq:outerface}
\area(f_o,D)
 = \sum_{f\in F'}\area(f,D).
 \end{equation}
\cref{eq:outerface} shows that for every face $f\in F'$, the equation systems $\AreaEqA {F'}$ and $\AreaEqA {F\setminus \{f\}}$  are equivalent. This fact is also used for \cref{thm:areaEqs}.
\begin{remark}
 In fact, \cref{lem:ccPlanar} and \cref{thm:areaEqs} generalize to \emph{inner triangulations}, i.e., $2$-connected plane graphs where every inner face is a triangle.
\end{remark}

%% file: 3polynomialMethod2.tex
\newpage
\section{Area-Universality of Triangulations with \pOrder{}s}\label{sec:poly}

We consider planar triangulations with the following property: An order of the vertices $(v_1,v_2,\dots,v_n)$, together with a  set of \emph{predecessors} $\pred(v_i)\subset N(v_i)$ for each vertex $v_i$, is a \emph{\pOrder} if the following conditions are satisfied:\smallskip
\begin{compactitem}[\ --]
 \item $\pred(v_i)\subseteq \{v_1,v_2,\dots, v_{i-1}\}$, i.e., the predecessors of $v_i$ have an index  $<i$,
 \item $\pred(v_1)=\emptyset$, $\pred(v_2)=\{v_1\}$, $\pred(v_3)=\pred(v_4)=\{v_1,v_2\}$, 
 and
 \item for all $i>4$: $|\pred(v_i)|=3$, i.e., $v_i$ has exactly three predecessors.
\end{compactitem}
\smallskip
Note that $\pred(v_i)$ specifies a subset of preceding neighbors.
Moreover, a \pOrder is defined for a planar graph independent of a drawing. 
We  usually denote a \pOrder by $\Po$ and state the order of the vertices; the predecessors are then implicitly given by $\pred(v_i)$. \cref{fig:3orient} illustrates a \pOrder.

\begin{figure}[hbt]
 \centering
 \includegraphics[page=2]{Figures}
 \caption{A plane 4-connected triangulation with a \pOrder \Po. 
 In an \arVP constructed with \Po, all  face areas are realized except for the two faces incident to the unoriented (dashed) edge $e_\Po$ of  $\Or_\Po$ (\cref{lem:almostReal}).}
 \label{fig:3orient}
\end{figure}

We pursue the following one-degree-of-freedom mechanism to construct realizing drawings for a plane triangulation $T$ with a \pOrder $(v_1,v_2,\dots,v_n)$ and an area assignment~\A:
\begin{compactitem}
 \item Place the vertices $v_1$, $v_2$, $v_3$ at positions realizing the area equation of the face $v_1v_2v_3$. Without loss of generality, we set $v_1=(0,0)$ and $v_2=(1,0)$.
 \item Insert $v_4$ such that the area equation of face $v_1v_2v_4$ is realized;
this is fulfilled if $y_4$ equals $\A(v_1v_2v_4)$ while $x_4\in \mathbb{R}$ is arbitrary. The value $x_4$ is our variable.
 \item Place each remaining vertex $v_i$ with respect to its predecessors $\pred(v_i)$ such that the area equations of the two incident face areas are respected; the coordinates of $v_i$ are \emph{rational} functions of $x_4$.
 \item Finally, all area equations are realized except for two special faces $f_a$ and~$f_b$.
Moreover, the face area of $f_a$ is a \emph{rational} function $\f$ of $x_4$.
 \item If $\f$ is \emph{almost surjective}, then there is a vertex placement~$D$ respecting all face areas and  orientations, i.e., $D$ is a real solution of $\AreaEqA{F}$.  
 \item By \cref{thm:areaEqs}, $D$ guarantees the realizability of $\A$.
 \item If this holds for \emph{enough} area assignments, then $T$ is area-universal.
\end{compactitem}

\subsection{Properties of \pOrder{}s}
A \pOrder \Po of a plane triangulation $T$ induces an orientation $\mathcal O_\Po$ of the edges: For $w\in\pred(v_i)$, we orient the edge from $v_i$ to $w$, see also \cref{fig:3orient}. 
By \cref{lem:4connnected}, we may restrict our attention to 4-connected triangulations. We note that  4-connectedness is not essential for our method but yields a cleaner picture. 
\begin{restatable}{lemma}{obsEdge}\label{obs:edge}
 Let $T$ be a planar 4-connected triangulation with a \pOrder \Po. Then  $\Or_\Po$ is acyclic, $\Or_\Po$ has a unique unoriented edge $e_\Po$, and $e_\Po$  
 is incident to $v_n$. 
\end{restatable}
%
%
It follows that the \pOrder encodes all but one edge which is easy to recover. Therefore, the 
  \pOrder of a planar triangulation $T$  encodes $T$. In fact, $T$ has a \pOrder if and only if there exists an edge $e$ such that $T-e$ is 3-degenerate.
%

\paragraph{Convention.}
Recall that a drawing induces an orientation of each face.
We follow the convention of stating the vertices of inner faces  ccw and of the outer face in cw direction.
This convention enables us to switch between different plane graphs of the same planar graph without changing the order of the vertices. 
To account for our convention,  we redefine $\A(f_o):=-\Sigma \A$ for the outer face~$f_o$. 
Then, for different embeddings, only the right sides of the \textsc{aeq}s change.
\smallskip

The next properties can be proved by induction and are shown in \cref{fig:predecessors}. 
\begin{restatable}{lemma}{construction}\label{lem:PolyMethodInsertion}
Let $T$ be a plane 4-connected triangulation with a  \pOrder \Po specified by $(v_1,v_2,\dots,v_n)$ and let $T_i$ denote the subgraph of $T$ induced by $\{v_1,v_2,\dots,v_i\}$. For $i\geq 4$,
\begin{compactitem}
 \item $T_i$ has one 4-face and otherwise only triangles,
 \item $T_{i+1}$ can be constructed from $T_{i}$ by inserting $v_{i+1}$ in the 4-face of $T_{i}$, and
 \item the three predecessors of $v_i$ can be named $(p_\textsc{f},p_\textsc{m},p_\textsc{l})$ such that $p_\textsc{f}p_\textsc{m}v_i$ and $p_\textsc{m}p_\textsc{l}v_i$ are (ccw inner and cw outer) faces of $T_i$.
\end{compactitem}
\end{restatable}

\begin{figure}[bh]
\begin{minipage}[b]{.68\textwidth}
\centering
\includegraphics[page=3]{Figures}
 \caption{
 Illustration of \cref{lem:PolyMethodInsertion}: (a) $T_4$, (b) $v_i$ is inserted in an inner 4-face, (c) $v_i$ is inserted in  outer 4-face.
}
\label{fig:predecessors}
\end{minipage}
\hfill
\begin{minipage}[b]{.225\textwidth}
\centering
 \includegraphics[page=4]{Figures}
 \caption{Illustration of \cref{obs:uniquePos}.}
 \label{fig:uniquePos}
\end{minipage}
\end{figure}

\begin{remark}
 For every (non-equivalent) plane graph $T'$ of $T$, the three predecessors $(p_\textsc{f},p_\textsc{m},p_\textsc{l})$ of $v_i$ in $T'$ and $T$ coincide.
\end{remark}

\begin{remark}\label{rmk:numberTriang} \cref{lem:PolyMethodInsertion} can be used to show that  
the number of 4-connected planar triangulations on $n$ vertices with a \pOrder is $ \Omega( \nicefrac{2^{n}}{n})$.
\end{remark}

\newpage
\subsection{Constructing Almost Realizing Vertex Placements}
Let $T$ be a plane triangulation with an area assignment $\mathcal A$. We call a vertex placement $D$ of $T$ \emph{almost $\A$-realizing}  if there exist two faces $f_a$ and $f_b$ such that~$D$ is a real solution of the equation system  $\AreaEqA{\tilde F}$ with $\tilde F:=F\setminus\{f_a,f_b\}$. In particular, we insist that the orientation and area of each face, except for $f_a$ and $f_b$ be correct, i.e., the area equations are fulfilled. Note that an \arVP does not necessarily correspond to a crossing-free  drawing.  
\begin{obs}\label{cor:THEone}
An almost $\A$-realizing vertex placement $D$ fulfilling the area equations of all faces except for $f_a$ and $f_b$, certifies the realizability of $\A$ if additionally the area equation of $f_a$ is satisfied.
\end{obs}

We construct \arVP{}s with the following lemma. 
 \begin{lemma}\label{obs:uniquePos}
Let $a,b\geq 0$ and let $q_\textsc{f},q_\textsc{m},q_\textsc{l}$ be three vertices with a non-collinear placement in the plane. Then there exists a unique placement for vertex $v$ such the ccw triangles  $q_\textsc{f}q_\textsc{m}v$ and $q_\textsc{m}q_\textsc{l}v$ fulfill the area equations for $a$ and $b$, respectively. 
\end{lemma}
\begin{proof}
Consider \cref{fig:uniquePos}.
To realize the areas, $v$ must be placed on a specific line $\ell_a$ and $\ell_b$, respectively.
Note that $\ell_a$ is parallel to the segment $q_\textsc{f},q_\textsc{m}$ and $\ell_b$ is parallel to the segment $q_\textsc{m},q_\textsc{l}$. Consequently, $\ell_a$ and $\ell_b$ are not parallel and their intersection point yields the unique position for vertex $v$. 
The coordinates of $v$ are specified by the two equations $\Det(q_\textsc{f},q_\textsc{m},v)\stackrel{!}{=}a$ and $\Det(q_\textsc{m},q_\textsc{l},v)\stackrel{!}{=}b$.
\end{proof}

Note that if $\ell_a$ and $\ell_b$ are parallel and do not coincide, then there is no position for $v$ realizing the area equations of the two triangles.
Based on \cref{obs:uniquePos}, we obtain our key lemma.
\begin{restatable}{lemma}{almostReal}\label{lem:almostReal}
Let $T$ be a plane 4-connected triangulation with a \pOrder \Po specified by $(v_1,v_2,\dots,v_n)$. Let $f_a,f_b$ be the faces incident to $e_\Po$ and $f_0:=v_1v_2v_3$. 
Then there exists a constant $c>0$ such that for a dense subset $\As_D$ of $\As^c$, every~$\A \in \As_D$ has
a finite set $\B(\A)\subset\R$, rational functions 
$x_i(\cdot,\A)$, $y_i(\cdot,\A)$, $\mathfrak f(\cdot,\A)$ and a triangle $\triangle$,
such that for all $x_4\in\R\setminus\B(\A)$, 
there exists a vertex placement $D(x_4)$ with the following properties:\smallskip
\begin{compactenum}[(i)]
\item $f_0$ coincides with the triangle $\triangle$, 
\item  $D(x_4)$ is almost realizing, i.e., a real solution of $\AreaEqA{F\setminus\{f_a,f_b\}}$,
\item every vertex $v_i$ is placed at the point $\big(x_i(x_4,\A),y_i(x_4,\A)\big)$, and
\item the area of face $f_a$ in $D(x_4)$ is given by $\f(x_4,\A)$.
 \end{compactenum}
\end{restatable}


The idea of the proof is to use \cref{obs:uniquePos} in order to construct $D(x_4)$ inductively. Therefore, given a vertex placement $v_1, \dots, v_{i-1}$, we have to ensure that the vertices of   $\pred(v_i)$ are not collinear.
To do so, we consider algebraically independent area assignments.  
We say an area assignment $\A$ of $T$ is \emph{algebraically independent} if 
the set  $\{\A(f)| f\in F'\}$ is algebraically independent over $\Q$. In fact, the subset of algebraically independent area assignments $\As_I$ of $\As^c$ is dense when $c$ is transcendental.

We call the function $\f$, constructed in the proof of \cref{lem:almostReal},  the \emph{\lastfaceFun} of $T$ and interpret it as a function in $x_4$ whose coefficients depend on~$\A$.

\subsection{Almost Surjectivity and Area-Universality}
In the following, we show that almost surjectivity of the \lastfaceFun implies area-universality. 
Let $A$ and $B$ be sets. A function $f\colon A\to B$ is \emph{almost surjective} if $f$ attains all but finitely many values  of $B$, i.e., $B\setminus f(A)$ is finite.

\begin{theorem}\label{thm:area-univ}
Let $T$ be a 4-connected plane triangulation with a \pOrder \Po and let $\As_D, \As^c,\f$ be obtained by \cref{lem:almostReal}. 
If the \lastfaceFun $\f$  is almost surjective for all area assignments in $\As_D$, then $T$ is area-universal.
\end{theorem}

\begin{proof}
By \cref{lem:fixedTotal}, it suffices to show that every  $\A\in\As_D$ is realizable. 
Let $f_0$ be the triangle formed by $v_1,v_2,v_3$ and $\As^+:=\As^{\leq 2c}|_{f_0\to \A(f_0)}$.
By \cref{cor:fillinSTRONGER}, $\A$ is realizable if  every open neighborhood of $\A$ in $\As^+$  contains a realizable area assignment.
Let $f_a$ and $f_b$ denote the faces incident to $e_\Po$ and $a:=\A(f_a)$.
\cref{lem:almostReal}  guarantees the existence of a finite set $\mathcal B$ such that 
for all $x_4\in \R\setminus\B$, there exists an almost $\A$-realizing vertex placement $D(x_4)$.
Since $\mathcal B$ is finite and $\f$ is almost surjective, for every $\varepsilon$ with $0<\varepsilon<c$, there exists  $\tilde x \in \R\setminus \mathcal B$ such that $a\leq\f(\tilde x)\leq a+\varepsilon$, i.e., the area of face $f_a$ in $D(\tilde x)$  is between $a$ and $a+\varepsilon$.
(If $f_a$ and $f_b$ are both inner faces, then the face $f_b$ has an area between  
$b-\varepsilon$ and $b$, where $b:=\A(f_b)$. Otherwise, if $f_a$ or $f_b$ is the outer face, then the total area changes and  face $f_b$  has area between $b$ and $b+\varepsilon$.)
Consequently, for some $\A'$ in the $\varepsilon$-neighborhood of $\A$ in $\As^+$, $D(\tilde x)$ is a real solution of $\AreaEq{\A'}{F\setminus\{f_b\}}$ and \cref{thm:areaEqs} ensures that $\A'$ is realizable.
By \cref{cor:fillinSTRONGER}, $\A$ is realizable. Thus, $T$ is area-universal.
\end{proof}

To prove area-universality, we use the following sufficient condition for almost surjectivity.  We say two real polynomials $p$ and $q$ are \emph{\crr} if they do not have common real roots. For a rational function $f:=\frac{p}{q}$, we define the \emph{max-degree} of $f$ as $\max\{|p|, |q|\}$, where $|p|$ denotes the degree of $p$. Moreover, we say $f$ is \crr if $p$ and $q$ are. 
The following property follows from the fact that polynomials of odd degree are surjective. 
%

\begin{restatable}{lemma}{surjective}\label{lem:surjective}
 Let $p,q\colon\mathbb R\to\mathbb R$ be polynomials and let $Q$ be the set of the real roots of~$q$.
 If the polynomials  $p$ and $q$ are \crr and have odd max-degree, then the function $f\colon \mathbb R\backslash{Q}\to \mathbb R, \ f(x)=\frac{p(x)}{q(x)}$ is almost surjective.
\end{restatable}

For the final result,  we make use of several convenient properties of algebraically independent area assignments.  For $\A$, let $\f_\A$ denote the last face function and $d_1(\f_\A)$ and $d_2(\f_\A)$ the degree of the numerator and denominator polynomial of $\f_\A$ in $x_4$, respectively. Since $\f_\A$ is a function in $x_4$ whose coefficients depend on $\A$, algebraic independence directly yields the following property.

\begin{restatable}{clm}{sameDegree}\label{obs:sameDegreeM}
For two algebraically independent area assignments $\A,\A'\in \As_I$ of a 4-connected triangulation with a $\pOrder$ $\Po$, the degrees of the last face functions $\f_\A$ and $\f_{\A'}$ with respect to $\Po$ coincide, i.e.,  $d_i(\f_\A)=d_i(\f_{\A'})$ for $i\in[2]$. 
\end{restatable}
In fact, the degrees do not only coincide for all algebraically independent area assignments, but also for different embeddings of the plane graph.
For a plane triangulation $T$, let $T^*$ denote the corresponding planar graph and~$[T]$ the set (of equivalence classes) of all plane graphs of $T^*$.

\begin{restatable}{clm}{sameDeg}\label{prop:sameDegM}
Let $T$ be a plane $4$-connected triangulation with a  \pOrder $\Po$. Then for every plane graph $T'\in[T]$, and algebraically independent area assignments $\A$ of $T$ and $\A'$ of $T'$, the last face functions $\f_\A$ and $\f'_{\A'}$  with respect to $\Po$ have the same degrees, i.e., $d_i(\f_\A)=d_i(\f'_{\A'})$ for $i\in[2]$.
\end{restatable}

This implies our final result: 
\begin{restatable}{corollary}{everyEmbedding}\label{thm:everyEmbedding}
 Let $T$ be a plane triangulation with a \pOrder $\Po$. If the \lastfaceFun~$\f$ of $T$ is \crr and has odd max-degree for one algebraically independent area assignment, then every plane graph in $[T]$ is area-universal. 
\end{restatable}

%% file: 5Applications.tex
\section{Applications}\label{sec:applications}
We now use \cref{thm:area-univ} and \cref{thm:everyEmbedding} to prove area-universality of some classes of triangulations.
%
The considered graphs rely on an operation that we call \emph{\diam}. Consider the left image of \cref{fig:diam}. Let $G$ be a plane graph and let $e$ be an inner edge incident to two triangular faces that consist  of $e$ and the vertices $u_1$ and $u_2$, respectively. 
Applying a \emph{\diam of order $k$} on $e$ results in the graph $G'$ which is obtained from $G$ by subdividing edge $e$ with $k$ vertices, $v_1,\dots,v_k$, and inserting the edges $v_iu_j$  for all pairs $i\in[k]$ and $j\in[2]$.  \cref{fig:diam} illustrates a \diam on $e$ of order 3.

\begin{figure}[htb]
 \centering
 \includegraphics[page=5]{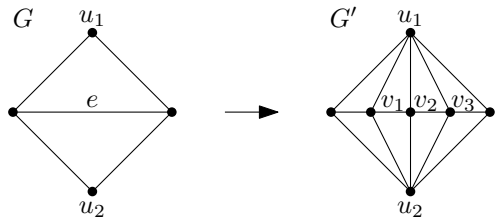}
 \caption{Obtaining $G'$ from $G$ by a \diam of order 3 on edge $e$.}
 \label{fig:diam}
\end{figure}

\subsection{\GraphClass graphs}
An \graphClass graph can be obtained from the plane octahedron graph $\mathcal G$ by a \diam{}: Choose one edge of the central triangle of $\mathcal G$ as the \emph{special edge}. 
The \emph{\graphClass} graph $\acc{\ell}$ is the plane graph obtained by a \diam of order $\ell$ on the special edge of $\mathcal G$.
Consequently, $\acc{\ell}$ has $\ell+6$ vertices.
We speak of an \emph{even \graphClass} if $\ell$ is even and of an \emph{odd \graphClass} if $\ell$ is odd.   \cref{fig:funnyClass2} illustrates the accordion graphs~$\acc{i}$ for $i\leq 3$. 
Note that  $\acc{0}$ is $\mathcal G$ itself and $\acc{1}$ is the unique 4-connected plane triangulation on seven vertices.
Due to its symmetry, it holds that $[\acc{\ell}]=\{\acc{\ell}\}$.
\begin{restatable}{theorem}{accordions}\label{thm:accordion}
The \graphClass\ graph \acc{\ell} is area-universal if and only if $\ell$ is odd.
\end{restatable}
\begin{proof}[Sketch]
Performing a \diam of order $\ell$ on some plane graph changes the degree of exactly two vertices by $\ell$ while all other vertex degrees remain the same. Consequently, if $\ell$ is even, all vertices of $\acc{\ell}$ have even degree, and hence,  $\acc{\ell}$ as an Eulerian triangulation is not area-universal as shown by the author in \cite[Theorem 1]{kleist1Journal}. 

It remains to prove the  area-universality of odd \graphClass graphs with the help of \cref{thm:area-univ}. Consider an arbitrary but fixed algebraically independent  area assignment $\A$.
We use the \pOrder depicted in \cref{fig:acc} to construct an \arVP.
We place the vertices $v_1$ at $(0,0)$, $v_2$ at $(1,0)$, $v_3$ at $(1,\Sigma\A)$, and $v_4$ at $(x_4,a)$ with $a:=\A(v_1v_2v_4)$. Consider also \cref{fig:acc}. 

\begin{figure}[htb]
 \centering
 \includegraphics[page=6]{Figures}
 \caption{ A \pOrder  of an  \graphClass graph (left) and an \arVP (right), where the shaded faces are realized.
}
 \label{fig:acc}
\end{figure}

We use \cref{lem:almostReal} to construct an \arVP.
Note that for all vertices $v_i$ with $i>5$, the three predecessors of $v_i$ are $p_\textsc{f}=v_3$, $p_\textsc{m}=v_{i-1}$ and $ p_\textsc{l}=v_4.$ One can show that the vertex coordinates of $v_i$ can be expressed as $x_i=\nicefrac{\Nx_i}{\D_i}$ and $y_i=\nicefrac{\Ny_i}{\D_i}$, where $\Nx_i,\Ny_i,\D_i$ are polynomials in~$x_4$.
Moreover, the polynomials fulfill the following crucial properties.
\begin{restatable}{lemma}{accDegree}\label{lem:accDegree}
For all $i\geq 5$, it holds that $|\D_5|=1$ and
\[|\Nx_{i+1}|=|\D_{i+1}|=|\Ny_{i+1}|+1=|\D_{i}|+1.\]
\end{restatable}

Consequently, $|\Nx_i|=|\D_i|$  is odd if and only if $i$ is odd. In particular, for odd~$\ell$, $|\Nx_n|=|\D_n|$ is odd  since the number of vertices  $n=\ell+6$ is odd.
\begin{restatable}{lemma}{accCRR}\label{lem:accCRR}
For all $i\geq 5$ and $\circ\in\{x,y\}$, it holds that 
$\N^\circ_i$ and $\D_i$ are \crr.
\end{restatable}
Consequently, the area of the ccw triangle $v_2v_3v_n$ in $D(x_4)$ is given by the \crr \lastfaceFun
\[\f(x):= \Det(v_2,v_3,v_n)=\Sigma\A(1-x_n)=\Sigma\A\left(1-\frac{\Nx_n}{\D_n}\right)
.\]
Since $|\Nx_n|$ and $|\D_n|$ are odd, the max-degree of $\f$ is odd. Thus, \cref{lem:surjective} ensures that  $\f$ is almost surjective. By \cref{thm:area-univ}, $\acc{\ell}$ is area-universal for odd $\ell$.
\end{proof}
 
 

This result can be generalized to double stacking graphs.

\subsection{Double Stacking Graphs}
Denote the vertices of the plane octahedron $\mathcal G$ by $ABC$ and $uvw$ as depicted in \cref{fig:doubleStacking}. 
The \emph{double stacking graph} \dS{\ell}{k} is the plane graph  obtained from~$\mathcal G$ by applying a \diam{} of order $\ell-1$ on $Au$ and a \diam of order $k-1$ on $vw$. 
Note that \dS{\ell}{k} has $(\ell+k+4)$ vertices. Moreover, \dS{\ell}{1} is isomorphic to \acc{\ell-1}; in particular, \dS{1}{1} equals $\mathcal G$. 
Note that $[\dS{\ell}{k}]$ usually  contains several (equivalence classes of) plane graphs.
\begin{figure}[t b]
 \centering
 \includegraphics[page=7]{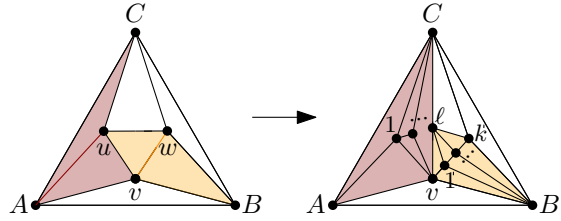}
 \caption{A double stacking graph $\dS{\ell}{k}$.}
 \label{fig:doubleStacking}
\end{figure}

\begin{restatable}{theorem}{doubleStacking}\label{thm:LKdoubleStacking}
A plane graph in $[\dS{\ell}{k}]$ is area-universal if and only if  $\ell \cdot k$ is even.
\end{restatable}
If $\ell \cdot k$ is odd, every plane graph in $[\dS{\ell}{k}]$ is Eulerian and hence not area-universal by 
\cite[Theorem 1]{kleist1Journal}.
If $\ell \cdot k$ is even, we consider an algebraically independent area assignment of $\dS{\ell}{k}$, show that its  \lastfaceFun is \crr and has odd max-degree. Then we apply \cref{thm:everyEmbedding}.\smallskip

\cref{thm:LKdoubleStacking} implies that
\begin{corollary}\label{cor:universal}
 For every $n\geq 7$, there exists a  4-connected triangulation on $n$ vertices that is area-universal.
\end{corollary}

%

%% file: 8Appendix1.tex
\section{Proofs of \cref{sec:prelim}}\label{app:2}

\subsection{Proof of \cref{cor:fillinSTRONGER}}

\cref{cor:fillinSTRONGER} follows from the fact that set of vertex placements $\D^{\leq c}|_{f\to\vartriangle}$ of $T$ is compact, where $\D^{\leq c}|_{f\to\vartriangle}$
 denotes the set of crossing-free vertex placements where $f$ coincides with a fixed triangle $\vartriangle$ of positive area and additionally the total area in each drawing does not exceed $c$.

\begin{lemma}\label{lem:compact}
Let $T$ be a plane triangulation $T$, let $f$ be some face of $T$, and $c\in \R_{>0}$. Then the set of vertex placements $\D^{\leq c}|_{f\to\vartriangle}$ of $T$ is compact.
\end{lemma}
\begin{proof}
 First note that closedness follows from the fact that we allow for degenerate drawings.
 Second, we show that $\D^{\leq c}|_{f\to\vartriangle}$ is bounded. If $f$ is the outer face, then clearly all inner vertex coordinates are bounded by the coordinates of $\vartriangle$. Hence, it remains to consider the case that $f$ is an inner face. By assumption, $\vartriangle$ has positive area and hence three sides which are pairwise not parallel. 
 
 We show that all vertices lie inside a bounded region which consists of the intersection of three half spaces. 
For each side $s$ of $f$, we consider the line $\ell_s$ such that $s$ and any point of $\ell_s$ form a triangle of area  $c$ that intersects $f$. Let $H_s$ denote the half space defined by $\ell_s$ that contains $f$. 
 We claim that any drawing~$D$ in $\D^{\leq c}|_{f\to\vartriangle}$ lies in $H_s$. Suppose a vertex $v$ of $T$ lies outside $H_s$, then the triangle~$t'$ formed by $s$ and $v$ is contained in $D$ since the outer face is triangular and thus convex. However, the area of $t'$ exceeds $c$.
 Therefore all vertices lie within  the intersection of the three half spaces $H_s$.
 \end{proof}
 
 Now we are ready to prove the proposition.
\fillingStronger*

\begin{proof}
Suppose $\A\in\As$ is realizable, then clearly, $\A$ itself lies in every of its neighborhoods and serves as a certificate of a realizing drawing.

Suppose every open neighborhood of $\A$ in $\As^{\leq 2c}|_{f\to \A(f)}$ contains a realizable area assignment.
Hence, we may construct a sequence of realizable assignments~$(\A_i)_{i\in\mathbb N}$ converging to $\A$.
Since $\A_i(f)=\A(f)$ for all $i$, \cref{lem:outerFace} allows to pick an  $A_i$-realizing drawing $D_i$ such that the placement of $f$ coincides  with a fixed triangle $\vartriangle$ of area $\A(f)$, i.e., $D_i\in\D^{\leq 2c}|_{f\to\vartriangle}$. By \cref{lem:compact}, the sequence $(D_i)_{i\in\mathbb N}$ is bounded. Therefore, by the Bolzano-Weierstrass theorem, $(D_i)_{i\in\mathbb N}$ contains a converging subsequence with limit $D$. By the compactness, $D$ is contained in $\D^{\leq 2c}|_{f\to\vartriangle}$ and thus yields a crossing-free drawing of $G$.   Note that for every inner face $f'\in F'$ it holds that:
\[\area(f',D)=\lim_{i\to\infty} \area(f',D_i)=\lim_{i\to\infty} \A_i(f')=\A(f').\]
 Consequently, $D$ guarantees that $\A$ is realizable.
\end{proof}

%
%

\subsection{Proof of \cref{lem:4connnected}}

\connected*

\begin{proof}
 If $T$ is area-universal, then also every subgraph and consequently all its 4-connected components are area-universal.
 %

We prove the other direction by induction. For the induction base, note that on $n=3$ and $n=4$ vertices there exist unique triangulations, namely the complete graphs on three and four vertices, which are area-universal as stacked triangulations and have no 4-connected components.
 
For the induction step consider a triangulation $T$ on $n>4$ vertices. Now we use the fact that a triangulation on $n>4$ vertices is 4-connected if and only if it has no separating triangle.
If $T$ is 4-connected, then the statement is vacuous. So suppose~$T$ has a separating triangle $t$. Let $T_i$ denote the triangulation consisting of $t$ and its interior, and let $T_o$ denote the triangulation consisting of $t$ and its exterior, see also \cref{fig:4connected}. 

 \begin{figure}[ht]
 \centering
 \includegraphics[page=1,width=.9\textwidth]{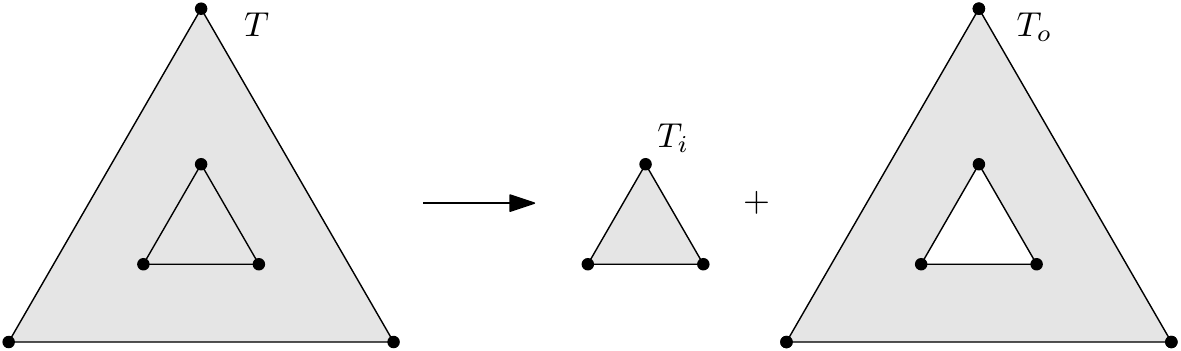}
 \caption{Decomposing a triangulation along separating triangles.}
 \label{fig:4connected}
\end{figure}
 By assumption all 4-connected components of $T$, and thus also of $T_i$ and $T_o$, are area-universal.  By induction it follows that  $T_i$ and $T_o$ are area-universal.
  We show how to obtain a realizing drawing of $T$ for any area assignment $\A$ of $T$. For $T_o$, consider a realizing drawing $D_o$ of $\A_o$ where $\A_o(t)$ equals the sum of all faces in $T_i$ and $\A_o(f)=\A(f)$ for all other faces of $T_o$. Now, for $T_i$ consider a realizing drawing $D_i$ of $\A_i$, the restriction of $\A$ to $T_i$. By \cref{lem:outerFace}, we may assume that the outer face of $D_i$ coincides with $t$ in $D_o$. By construction, $A_o(t)=\Sigma A_i$ and hence the area of $t$ in $D_o$  coincides with the area of the outer triangle of $T_i$. Hence, the union of $D_o$ and $D_i$ yields an $\A$-realizing drawing of $T$.
\end{proof}

\newpage
\subsection{Proof of \cref{thm:areaEqs}}
\propThree*
\begin{proof} 
The proof consists of two directions. 
If $\A$ is realizable, then the vertex placement of an $\nicefrac{1}{2}\A$-realizing drawing is a real solution of \AreaEqA{F\setminus \{f\}}. Recall that by \cref{lem:scaling}, $\A$ is realizable if $\nicefrac{1}{2}\A$ is.

If  \AreaEqA{F\setminus\{f\}} has a real solution $S$, $S$ yields a vertex placement $D$ satisfying $\A$ and preserving the orientation of all but one face $f$. It remains to show that $D$ corresponds to a crossing-free drawing. 
%
If $f$ is the outer face, then 
\cref{lem:ccPlanar} implies that $D$ is an equivalent straight-line drawing of $T$.
\cref{eq:outerface} shows that for every face $f'\in F'$, the equation systems $\AreaEqA {F'}$ and $\AreaEqA {F\setminus \{f'\}}$  are equivalent.
Consequently, we may assume that $f$  is the outer face $f_o$.
Thus, it remains to prove \cref{eq:outerface}. Let $v_1,v_2,\dots, v_k$ denote the vertices of the outer face $f_o$ (of an inner triangulation) in counter clockwise orientation. 
Recall that $\Det(u,v,w)$ denotes the determinant of the homogeneous coordinates of $u,v,w$ as defined in \cref{eq:det}. Moreover,  $\det(u,v)$ denotes the determinant of the 2-dimensional coordinates of $u$ and $v$.
Then, by the properties of the determinant, for any vertex placement $D$ it holds that
\begin{align*}
2\cdot \area(f_o,D)=&\sum_{i=2}^{k-1} \Det(v_1,v_i,v_{i+1})=\sum_{i=1}^{k-1} \det(v_i,v_{i+1})\\
=&\sum_{i=1}^{k-1} \det(v_i,v_{i+1})+\sum_{e=(u,v) \text{inner}} \big(\det(u,v)+\det(u,v)\big)\\
=& \sum_{f=(u,v,w)\in F'}\Det (u,v,w)
= 2\cdot \sum_{f\in F'}\area(f,D).
\end{align*}
In the last line, we use the fact that every inner edge appears in both direction and every outer edge in one direction. Thus we can traverse each inner face in ccw direction. This fact is illustrated in \cref{fig:neatLemma}.
\end{proof}

\begin{figure}[bth]
 \centering
 \includegraphics[page=8]{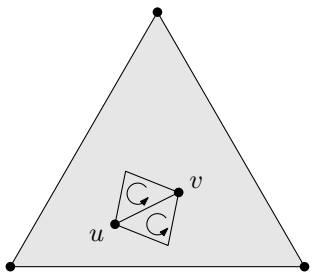}
 \caption{Illustration of the proof of \cref{thm:areaEqs}.}
 \label{fig:neatLemma}
\end{figure}

%% file: 8Appendix2.tex
\newpage
\section{Proofs of \cref{sec:poly}}\label{app:3}
Here we present the omitted proofs of \cref{sec:poly}. We start with the properties of \pOrder{}s.

\obsEdge*
\begin{proof}
 By definition of $\Or_\Po$ if an edge $(v_{i+1},v_j)$ is oriented from $v_{i+1}$ to $v_j$, then $i>j$. Hence the orientation is acyclic. In particular, no edge is oriented in two directions.
 The number of unoriented edges follows by double counting the edges of $T$. On the one hand, by Euler's formula, the number of edges in a triangulation is $|E|=3n-6$. On the other hand, the number of oriented edges $E_\uparrow$ is given by the sum of the outdegrees. 
 $$|E_\uparrow|=\sum_{i=1}^n \outdeg(v_{i+1})= 0+1+2+2+3(n-4)=3n-7.$$ 
 Hence, $|E|-|E_\uparrow|=1$ and thus there is exactly one edge $e$ without orientation. 
 Observe that the last vertex $v_n$ in the \pOrder
 has indegree 0. If $T$ has minimum degree 4, then $e$ is incident to $v_n$; otherwise $v_n$ would be a vertex of degree 3. This is a contradiction.
\end{proof}

\construction*
\begin{proof} 
We prove this statement by induction.  
For the induction base, note that~$T_4$ has three faces: the triangle $v_3v_2v_1$, the triangle $v_1v_2v_4$, and the 4-face $v_1v_4v_2v_3$. 
\cref{fig:predecessors}\,(a) depicts~$T_4$ for the case that $v_1v_2v_3$ is the outer face. By 4-connected\-ness, $v_3$ and $v_4$ cannot share an edge.  Therefore, with this notation, the inner faces are ccw and the outer face is cw oriented -- independent of the choice of the outer face.

Now we consider the induction step and insert $v_{i+1}$ in $T_i$. Since $T$ is 4-connected, $v_{i+1}$ can only be placed in the unique 4-face $f$ of $T_{i}$. Clearly, any three vertices of $f$ are consecutive on the boundary cycle of $f$.
Hence, the predecessors of $v_{i+1}$ form a path of length three along $f$. We define $p_\textsc{m}$ as  the middle vertex of this path.
Naming the remaining predecessors by  $p_\textsc{f}$ and $p_\textsc{l}$, $p_\textsc{f}p_\textsc{m}v_{i+1}$ and $p_\textsc{m}p_\textsc{l}v_{i+1}$ are (not necessarily correctly oriented) triangles in $T_{i+1}$. Since $T$ is 4-connected, these triangles of $T_{i+1}$ are faces in $T$ and thus also in $T_{i+1}$. Furthermore, $p_\textsc{f}v_{i+1}p_\textsc{l}w$ forms a 4-face of $T_{i+1}$ where $w$ is the vertex of $f$ which is not in $\pred(v_{i+1})$.

For the correct orientation we distinguish two cases: If $f$ is an inner face, we define $p_\textsc{f}$ as the ccw first vertex (of the path of predecessors in $f$) and $p_\textsc{l}$ as the ccw last vertex. \cref{fig:predecessors}\,(b) illustrates this definition for the case that $f$ is an inner face.
Otherwise, $f$ is the outer face and we define $p_\textsc{f}$  as the cw first vertex  and $p_\textsc{l}$ the cw last vertex. This case is displayed in \cref{fig:predecessors}\,(c). Then,  $p_\textsc{f}p_\textsc{m}v_{i+1}$ and $p_\textsc{m}p_\textsc{l}v_{i+1}$ are ccw faces in $T_{i+1}$ if and only if they are inner faces of~$T$.
\end{proof}

As mentioned in \cref{rmk:numberTriang}, \cref{lem:PolyMethodInsertion} can be used to obtain a lower bound on the number of 4-connected planar triangulations on $n$ vertices with a \pOrder.

\begin{proposition}
The number of 4-connected planar triangulations on $n$ vertices with a \pOrder is $ \Omega( \nicefrac{1}{n}\cdot 2^{n})$.
\end{proposition}
\begin{proof}
Firstly, we show that a 4-connected triangulation $T$ on $n$ vertices has at most $9n\cdot 2^{n}$ different \pOrder{}s.
We consider every \pOrder in the reverse order $v_n,...,v_1$. 
$T$ has at most $3n$ edges which may serve as the unique unoriented edge. Its deletion yields a 4-face. By \cref{lem:PolyMethodInsertion}, in every \pOrder for $i=n,...,5$ the vertex $v_i$ is a vertex incident to a 4-face in~$T_i$. Removing $v_i$ from $T_i$ yields a graph $T_{i-1}$ that has again a unique 4-face. It follows from \cref{obs:edge} that $v_i$ is a vertex of degree 3 in $T_i$. Consequently, all neighbors of $v_i$ in $T_i$ are the predecessors of $v_i$.
\vspace{-15pt}
\begin{figure}[h!]
\begin{minipage}{.58\textwidth}
\normalsize
Consider \cref{fig:vertdeg3} and observe that two adjacent vertices of degree 3 on a 4-face $x_1y_1y_2x_2$ certify a separating triangle, unless $n\leq 5$: Since all other faces of $T_i$ are triangles, every pair of adjacent vertices has a common neighbor outside of the $4$-face. If $x_i$  has degree 3, the common neighbor of $x_iy_i$ and $x_1x_2$ coincides; we call it $z$. Thus, $zy_1y_2$ is a separating triangle unless $T_i$ contains only these five vertices. By 4-connectivity, each 4-face has at most two vertices of degree 3 for $i\geq 5$ and there are at most two choices for the vertex $v_i$. For $i=5$, the number of choices is upper bounded by the four vertices of the 4-cycle and for $i\leq 4$, by another four; two for $v_1,v_2$ and two for $v_3,v_4$.
\end{minipage}
\hfill
\begin{minipage}{.35\textwidth}
\centering
\includegraphics[page=4]{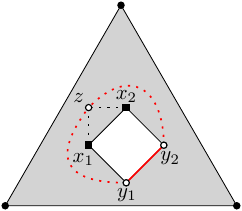}
\caption{Two adjacent vertices of degree three in a 4-face certify a separating triangle.}
\label{fig:vertdeg3}
\end{minipage}
 
 \end{figure}

\vspace{-18pt}
Consequently, for a specific unoriented edge, the number of vertex orderings is at most $2^{n-5}\cdot 4 \cdot 4=2^{n-1}$. This makes a total of at most $3n\cdot2^{n-1}$ different \pOrder{}s for a fixed triangulation.

In order to build a 4-connected triangulation with a \pOrder $v_1,v_2,\dots, v_n$, we specify the middle predecessor $v_\textsc{m}$ of $v_i$ for $5\leq i\leq n$ from the 4-face of $T_{i-1}$. By \cref{lem:PolyMethodInsertion}, the remaining two predecessors of $v_i$ are the two neighbors of $v_\textsc{m}$ in the 4-face.
Thus, we have four choices for $v_\textsc{m}$ in each step $i>5$. For $i=5$, neither $v_3$ nor $v_4$ can serve as the middle predecessor since this results in a separating triangle. Thus, we obtain at least $2\cdot 4^{n-5}$ different \pOrder{}s. By the above observation at most $3n\cdot2^{n-1}$ belong to the same triangulation. 
 Hence there exist $\Omega(\nicefrac{1}{n}\cdot 2^{n})$ 4-connected planar triangulation on $n$ vertices.
\end{proof}

\subsection{Proof of \cref{lem:almostReal}}\label{app:Lemma}
For the proof of \cref{lem:almostReal}, we introduce the concept of algebraically independent area assignments.
A set of real numbers $\{a_1,a_2,\dots,a_k\}$ is \emph{algebraically independent} over $\Q$ if for each polynomial $p(x_1,x_2,\dots,x_k)$ with coefficients from $\Q$, different from the 0-polynomial, it holds that $p(a_1,a_2,\dots,a_k)\neq 0$. 
We say a face area assignment $\A$ of a plane graph $G$ is \emph{algebraically independent} if 
the set  $\{\A(f)| f\in F'\}$ is algebraically independent over $\Q$. Note that then for all $f_0\in F'$ it holds that  $\{\A(f)| f\in F'\setminus f_0\}\cup \{\Sigma \A\}$ is algebraically independent.
For a transcendental $c$, we denote the subset of $\As^c$ consisting  of all algebraically independent area assignments  by $\As_I$ and show that $\As_I$ is a dense subset.
\begin{clm}\label{lem:dense}
If $c$ is transcendental $c$, then $\As_I$ is dense in  $\As^c$.
\end{clm}

\begin{proof}
We show by induction that the set of algebraically independent $k$-tuples is dense.
Our proof is built upon the fact that  the algebraic closure of a countable field is countable \cite[p. 343, Cor. B-2.41]{rotmanAlgebra}.

For the induction base, we consider $k=1$. Since the algebraic closure of~$\Q$ is countable, its complement $\R\setminus \Q$ is dense in $\R$. 

Now we consider the induction step from $k-1$ to $k$. In order to show that the algebraically independent $k$-tuple are dense, it suffices to show that for each $a$ in $\R^n$ and each $\varepsilon$-ball $B$ of $a$, there exists an algebraically independent $b$ in $B$.
By induction hypothesis, we find a $(b_2,\dots,b_k)$ with algebraically independent entries, that is arbitrarily close to $(a_2,\dots,a_k)$. Let $K$ denote the algebraic closure of $\Q(b_2,...,b_n)$, the smallest field containing $\Q$ and $\{b_2,\dots,b_k\}$. As a rational function field over $\Q$, the field $\Q(b_2,...,b_n)$  is countable. 
Thus, since the algebraic closure of a countable field is countable, $K$ is countable. Thus in each open neighborhood of $a_1$, there exists a $b_1$ in the complement of $K$. Therefore, in each $\varepsilon$-ball of $a=(a_1,\dots,a_n)$ there exists an algebraically independent $b=(b_1,\dots,b_n)$. 
\end{proof}

Now, we are ready to prove the main lemma.
\almostReal*
\begin{proof}
Let $c$ be transcendental. We show that the claim holds for $\As_D:=\As_I$.
Thus, we consider an arbitrary $\A\in \As_I$ and think of $\A$ as an \emph{abstract} area assignment, where the prescribed areas are still variables. 
The idea of the proof is simple. Given a placement of $v_1, \dots, v_{i-1}$, we want to insert $v_i$ by \cref{obs:uniquePos}. Thus, we need to guarantee that the predecessors are not collinear.

We rename the vertices $v_1,v_2,v_3,v_4$ such that the triangle $v_1v_3v_2$ and $v_1v_2v_4$ are ccw inner and cw outer faces.
For simplicity, we can think of $f_0$ as being the outer face. However, the construction works in all settings.

We place $v_1$ at $(0,0)$, $v_2$ at $(0,1)$ and set $y_3:=-\A(f_0)$; recall that for the outer face $f_o$ it holds that $\A(f_o):=-\Sigma \A$. We use the freedom to specify $x_3$ at a later state.  
This guarantees property (i).
Furthermore, for $a:=\A(v_1v_2v_4)$ we place $v_4$ at $(x_4,a)$. Consequently, the face area of the triangle $v_1v_2v_4$ is realized  for all choices of $x_4$.


For property (ii), we show that for all but finitely many values of $x_4$,
we obtain an almost $\A$-realizing vertex placement $D(x_4)$. 
By \cref{lem:PolyMethodInsertion}, the three predecessors of $v_i$ can be named $(p_\textsc{f},p_\textsc{m},p_\textsc{l})$ such that 
$p_\textsc{f}p_\textsc{m}v_i$ and $p_\textsc{m}p_\textsc{l}v_i$ are ccw inner or cw outer faces of $T$.
By \cref{obs:uniquePos}, the vertex coordinates of $v_i$ follow directly from its three predecessors $p_\textsc{f},p_\textsc{m},p_\textsc{l}$ -- unless these are collinear.  
Denoting the coordinates of $p_i$ by $(x_i,y_i)$  and solving the area equations yields the following coordinates of $v_i$:
\begin{linenomath}
\begin{align}\label{eq:1}
 x_i&=x_\textsc{m}+\frac{a(x_\textsc{l} - x_\textsc{m})-b(x_\textsc{m} - x_\textsc{f})}
    { x_\textsc{f} (y_\textsc{m}-y_\textsc{l} )+x_\textsc{m}(y_\textsc{l} -y_\textsc{f}) + x_\textsc{l} (y_\textsc{f}-y_\textsc{m})}
\end{align}
\begin{align}\label{eq:2}
y_i&=y_\textsc{m}+\frac{ a(y_\textsc{l} - y_\textsc{m})-b(y_\textsc{m} - y_\textsc{f})}
    { x_\textsc{f} (y_\textsc{m}-y_\textsc{l} )+x_\textsc{m}(y_\textsc{l} -y_\textsc{f}) + x_\textsc{l} (y_\textsc{f}-y_\textsc{m})}
\end{align}
\end{linenomath}
Note that the predecessors are collinear if and only if the denominators of \cref{eq:1,eq:2} vanish.

Assume for now, that we are considering a  position for $x_4$ such that no triple of predecessors becomes collinear. For $i=5,\dots,n$, we place $v_i$ according to \cref{obs:uniquePos} and satisfy two new area equations.
Together with the realized face area of the triangles $v_1v_2v_4$ and $v_1v_2v_3$, the number of realized face areas is  
\[2(n-4)+2=2n-6.\]
Consequently, all but two face areas, namely $f_a$ and $f_b$, are realized and $D(x_4)$ is an \arVP.
Let $\B(\A)$ denote the set of all $x_4$ where a triple of predecessors becomes collinear.
We postpone to show that $\B(\A)$ is finite. It is sufficient to show that the denominator of each vertex is not the 0-polynomial. We prove this simultaneously with property (iii).
\smallskip

Now, we show (ii). For each vertex $v_i$, we wish to represent its coordinates $(x_i,y_i)$ in $D(x_4)$ by rational functions with a common denominator. Specifically, we aim for polynomials $\Nx_i,\Ny_i,\D_i$ in $x_4$, which are different from the 0-polynomial, such that
\[x_i=\frac{\Nx_i}{\D_i} \text{\qquad\ and \ \qquad} y_i=\frac{\Ny_i}{\D_i}.\]
Moreover, we assume that the leading coefficient of $\D_i$ is 1.
We show the existence of such a representation by induction. Hence assume that we have such a representation of $v_1,\dots,v_{i-1}$. By the placement of the initial vertices, it holds that   $\D_i:=1$ for all $i\in[4]$.

Now, we consider the vertex $v_i$ with $i>4$.
By \cref{lem:PolyMethodInsertion}, we denote the three predecessors of $v_i$ by $p_\textsc{f},p_\textsc{m},p_\textsc{l}$ such that the triangles 
$p_\textsc{f}p_\textsc{m}v_i$ and $p_\textsc{m}p_\textsc{l}v_i$ are ccw inner or cw outer faces of $T$; we call the prescribed face areas of the two triangles $a_i$ and $b_i$, respectively. \cref{eq:1,eq:2}  yield the coordinates of vertex $v_i$ in $D(x_4)$. Since we will aim for the fact, that the representation is \emph{\crr}, that is the polynomials share no common real root, we are already here more careful. 

For the later argument, it is convenient to consider $v_5$ explicitly. 
By the 4-connectedness of $T$, neither $v_3$ nor $v_4$ is the middle predecessor of $v_5$. Thus, by symmetry, we may assume that   $\pred(v_5)=\{v_1,v_4,v_3\}$ and \cref{eq:1,eq:2} simplify to 
\begin{linenomath}
\begin{align}\label{eq:v5}
  x_5=\frac{a_5x_4+b_5x_3}{y_3x_4-ax_3} \qquad \text{ and } \qquad
  y_5=\frac{aa_5+b_5y_3}{y_3x_4-ax_3}. 
\end{align}
\end{linenomath}

Note that for $x_3=0$, the denominators of $x_5$ and $y_5$ would vary in a \crr representation. Thus, for $x_3\neq 0$, we define $\Nx_5:=x_4a_5+b_5x_3$, $\Ny_5:=aa_5+b_5$, $\D_5:=y_3x_4-ax_3$. Clearly, none of them is the 0-polynomial. 

Now, we consider the induction step for vertex $v_i$ with $i>5$. Note that, due to the 4-connectedness,  $v_1$ is not a predecessor of $v_i$.
\cref{eq:1,eq:2}  yield the coordinates of vertex $v_i$ in the \arVP. 
By assumption, $\D_\textsc{f}\cdot\D_\textsc{m}\cdot\D_\textsc{l}$ is not the 0-polynomial, since none of its factors is the 0-polynomial. Therefore, we may 
expand the right term by $\D_\textsc{f}\cdot\D_\textsc{m}\cdot\D_\textsc{l}$. Using the representations $x_j=\nicefrac{\Nx_j}{\D_j}$ and $y_j=\nicefrac{\Ny_j}{\D_j}$ for $j\in\{\textsc{f,m,l}\}$ yields the following identities:
\begin{linenomath}
\begin{align}
 x_i&=\frac{\Nx_\textsc{m}}{\D_\textsc{m}}+\frac{\D_\textsc{m}(a_i\Nx_\textsc{l}\D_\textsc{f}+b_i \Nx_\textsc{f}\D_\textsc{l}) - (a_i+b_i)\Nx_\textsc{m}\D_\textsc{f}\D_\textsc{l}}{\tilde \D_i}\label{eq:X}\\[6pt]
  y_i&=\frac{\Ny_\textsc{m}}{\D_\textsc{m}}+\frac{\D_\textsc{m}( a_i\Ny_\textsc{l}\D_\textsc{f}+b_i \Ny_\textsc{f}\D_\textsc{l} ) -(a_i+b_i)\Ny_\textsc{m}\D_\textsc{f}\D_\textsc{l}}%
    { \tilde \D_{i}}\label{eq:Y}
\end{align}
\begin{equation}\label{eq:tilde}
 \text{with }\tilde \D_{i}:=\Nx_\textsc{f} (\Ny_\textsc{m}\D_\textsc{l}-\Ny_\textsc{l}\D_\textsc{m} )+\Nx_\textsc{m}(\Ny_\textsc{l}\D_\textsc{f} -\Ny_\textsc{f}\D_\textsc{l}) + \Nx_\textsc{l} (\Ny_\textsc{f}\D_\textsc{m}-\Ny_\textsc{m}\D_\textsc{f}).
\end{equation}
\end{linenomath}
Note that the denominators of $x_i$ and $y_i$ are identical and the numerators are symmetric in the $x$- and $y$-coordinates of their predecessors, respectively. Hence, for $\circ\in\{x,y\}$, we define $\N^\circ_{i}$ to unify the notation.

We wish to argue that $\tilde \D_i$ is not the 0-polynomial. The existence of distinct $j,k\in\{\textsc{f,m,l}\}$ such that neither $\N^x_j$ nor $\N^y_k$ are the 0-polynomial guarantees that one summand of $\tilde \D_i$ does not vanish. Note that such a pair does always exist since there are only two polynomials which might be the 0-polynomial, namely $\Ny_2$ and $\Nx_3$; here we use the fact that  $v_1$ is no predecessor of any $v_i$ with $i>5$. 


We now expand to find the desired representations of $x_i$ and $y_i$. The denominator is the least common multiple of $\D_\textsc{m}$ and $\tilde\D_i$, none of which is the 0-polynomial. 
Thus, we define $E_i$ and $F_i$ to be \crr polynomials such that
\begin{linenomath}
\begin{equation}\label{eq:EF}
 \D_\textsc{m} E_i=\tilde \D_{i}F_i.
\end{equation}
\end{linenomath}
In \cref{eq:X,eq:Y}, we expand
the left summand by $E_i$ and the right summand by $F_i$. Then the coordinates of vertex $v_i$ can be expressed by 
\begin{linenomath}
\begin{align}\label{eq:NomDen}
 \N^\circ_i&:=F_i\D_\textsc{m}(a_i\N^\circ_\textsc{l}\D_\textsc{f}+b_i \N^\circ_\textsc{f}\D_\textsc{l}) + \N^\circ_\textsc{m}(E_i-(a_i+b_i)\D_\textsc{f}\D_\textsc{l}F_i)\\\label{eq:NomDen2}
 \D_i&:=\D_\textsc{m} E_i=\tilde \D_{i}F_i 
 \end{align}
\end{linenomath}
Thus, each coordinate of $v_i$ is a rational function in $x_4$, where the coefficients are polynomials in 
$\A$. Due to the algebraically independence of $\A$, the coefficients cannot vanish and $\D_i$ is not the 0-polynomial. Using the fact that  $\D_\textsc{f}\cdot\D_\textsc{m}\cdot\D_\textsc{l}$ is not the 0-polynomial, any $\N^\circ_j$ with $j\in\{\textsc{f,m,l}\}$ which is not the 0-polynomial, certifies that $\N_i^\circ$ is not the zero polynomial.  Recall that we already guaranteed the existence of distinct $j,k\in\{\textsc{f,m,l}\}$ such that neither $\N^x_j$ nor $\N^y_k$ are the 0-polynomial. Thus, such a pair also implies that $\N_i^\circ$ is not the zero polynomial for all choices of $\circ\in\{x,y\}$. Consequently, we have proved property (ii) and (iii).
\smallskip

Moreover, (iii) immediately implies (iv): The area of face $f_a$ can be expressed as the determinant of its three vertex coordinates. Thus, if the vertex coordinates are rational functions in $x_4$, so is  the area of face $f_a$.
\end{proof}

We interpret $\f$ as a rational function in $x_4$ whose coefficients depend on $\A$. 

\subsection{Almost surjectivity and area-universality}\label{app:B2}
We start by proving \cref{lem:surjective}.
\surjective*
\begin{proof}
 Let $c\in\mathbb R\backslash \{0\}$ and consider $g\colon \R\to\R, \ g(x):=p(x)-c q(x)$. The leading coefficients of the polynomials $p$ and $c\cdot q$ cancel for at most one choice of $c$. For all other values, the degree of $g$ is 
 $\deg g= \max\{\deg p, \deg q\}$ and by assumption odd.  
 Consequently, as a real polynomial of odd degree, $g$ has a real root $\tilde x$. If $q(\tilde x)\neq 0$, then \[g(\tilde x)=0\iff f(\tilde x)=c.\]
 Suppose $q(\tilde x)=0$. Then, $q(\tilde x)=0=g(\tilde x)=p(\tilde x)$. Hence $\tilde x$ is a zero of both $p$ and $q$. A contradiction to the assumption that $p$ and $q$ are \crr. 
\end{proof}


Now, we aim to prove \cref{thm:everyEmbedding} which relies on several interesting properties of algebraically independent area assignments. In particular, it remains to show \cref{prop:sameDegM}. 
%

\sameDeg*
\begin{proof}
We assume that $v_1v_2v_3$ is the cw outer face of $T$ and $w_1w_2w_3$  is the cw outer face of $T'$.
Then $v_1v_2v_3$ is a ccw inner face of $T'$ and $w_1w_2w_3$ is a ccw inner face of $T$.
Compared to $T$,  the orientation of the faces $v_1v_2v_3$ and $w_1w_2w_3$ in~$T'$ changes; while the orientation of all other faces remains:
This is easily seen when considering the drawings on the sphere; which are obtained by one-point compactification of some point in the respective outer faces. Due to the 3-connectedness and  Whitney’s uniqueness theorem, the drawings $T$ and $T'$ on the sphere are equivalent \cite{whitney1933-2}. Moreover, in the drawing of the sphere $v_1v_2v_3$ and $w_1w_2w_3$ are ccw. 
Then choosing one face as the outer face and  applying a stereographic projection of the punctured sphere where a point of the outer face is deleted, results in a drawing in the plane where 
all faces remain ccw while the outer face becomes cw.


With respect to the area assignment $\A$ of $T$,  $v_1v_2v_3$ is assigned to the total area $-\Sigma \A$ and $w_1w_2w_3$ to some value~$c$.
As an intermediate step we consider the area assignment~$\A''$ of $T'$ where all area assignments remain but 
$w_1w_2w_3$ obtains the total area $-\Sigma \A$ and $v_1v_2v_3$ some value $c$. Clearly, $\A''$ is algebraically independent since $\A$ is. Fortunately, the negative sign accounts for the fact that the orientation changes; constructing realizing drawings by \cref{lem:almostReal} $T$ and $T'$ are treated by the very same procedure. 
Thus, $\f'_{\A''}$ can be obtained from $\f_\A$ by swapping all occurrences of   $c$ and $-\Sigma \A$. 
Consequently, the degrees of the denominator and numerator polynomials of the last face functions $\f_\A$ and $\f'_{\A''}$ coincide. Moreover, by \cref{obs:sameDegreeM}, the degrees of $\f'_{\A''}$ and $\f'_{\A'}$ coincide.
\end{proof}

Consequently, \cref{lem:surjective}, \cref{prop:sameDegM}, and \cref{thm:area-univ} imply \cref{thm:everyEmbedding}.
\everyEmbedding*
\begin{proof}
If the  \lastfaceFun $\f$ of $T$ has odd max-degree for some $\A\in \As_I$, then this holds true for all area assignments in $\As_I$ by  \cref{obs:sameDegreeM}. Consequently, \cref{lem:surjective} guarantees that 
the \lastfaceFun $\f(\cdot,\A)$ is almost surjective for all $\A\in\As_I$. Since $\As_I$ is dense in $\As$ as proved in \cref{lem:dense}, \cref{thm:area-univ} implies the area-universality of $T$.

For every other plane graph $T'\in[T]$, the \lastfaceFun $\f'$ has also odd maximum degree by \cref{prop:sameDegM}. Here we used the fact that~$\f'$ can be obtained from~$\f$ by exchanging two algebraically independent numbers. By the same reasoning,  $\f'$ is also \crr. Thus, the above argument shows that $T'$ is area-universal.
\end{proof}

%% file: 8Appendix3.tex
\newpage
\section{Proofs of \cref{sec:applications}}\label{app:applications}
Here we present the omitted proofs of \cref{sec:applications}. 
We start by helpful lemmas to analyze the coordinate functions and their degrees.

\input{4CoordinateLemmasNotation}
\newpage
\subsection{Proof of \cref{thm:accordion}}\label{app:acc}
\accordions*

To complete the proof of \cref{thm:accordion}, we have to show \cref{lem:accDegree} and \cref{lem:accCRR}.
\accDegree*

\begin{proof} 
 We denote the two face areas incident to $v_i$ and its predecessors by $a_i$ and $b_i$ for $i\geq5$, see \cref{fig:funnyClass1} for an illustration.
\begin{figure}[htb]
 \centering
\begin{subfigure}[t]{.45\textwidth}
\centering
 \includegraphics[page=1,scale=.75]{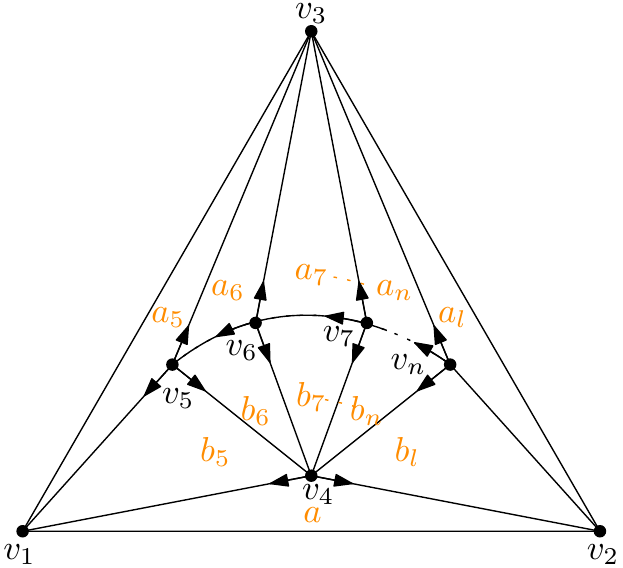}
 \caption{A \pOrder  of an  \graphClass graph.}
 \label{fig:funnyClassA}
\end{subfigure}
\hfill
\begin{subfigure}[t]{.45\textwidth}
\centering
 \includegraphics[page=4,scale=.75]{funnyClass}
 \caption{An \arVP; the green faces are realized.}
 \label{fig:funnyClassB}
\end{subfigure}
 \caption{Illustration of \cref{thm:accordion} and its proof.}
 \label{fig:funnyClass1}
\end{figure}
 
 We show this claim by induction.
 For $v_5$, we have already evaluated the \crr vertex coordinates in \cref{eq:v5}. 
Thus, it holds that 
\[|\Nx_{5}|=|\D_{5}|=|\Ny_{5}|+1=1.\]
For $\circ\in\{x,y\}$, we define $d_i^\circ:=|\N^\circ_i|-|\D_i|$. Consequently, it holds that $d_5^x=0$ and $d_5^y=-1$. Moreover, since $\D_\textsc{m}=1$, it holds that $F_5=1$ and $\D_5=E_5=\tilde \D_5$.
 
Recall that all vertices $v_i$ with $i>5$, the three predecessors of $v_i$ are $p_\textsc{f}=v_3, \quad p_\textsc{m}=v_{i-1}\quad \text{and} \quad p_\textsc{l}=v_4.$
In other words, all vertices $v_i$ with $i>4$ are stacked on the same angle.
Consequently, for all $i>5$, it holds that 
$|\N^\circ_\textsc{f}|=|\D_\textsc{f}|=|\D_\textsc{l}|=0$ and  $d^x_\textsc{f}=d^y_\textsc{f}=0$, and 
$d^x_\textsc{l}=1, d^y_\textsc{l}=0$. 
Defining $M:=\max\{|E_5|,|\D_\textsc{f}|+|\D_\textsc{l}|+|F_5|\}=\max\{1,0\}=1$, we obtain for $v_i$, $i>5$, with the help of \cref{lem:degreesStacked}:
\begin{linenomath}
\begin{align*}
|\Nx_{i+1}|&=|\D_{i}|+M+ d^x_{5}=|\D_{i}|+1\\
|\Ny_{i+1}|&=|\D_{i}|+M+ d^y_{5}=|\D_{i}|\\
|\D_{i+1}|&=|\D_{i}|+M=|\D_{i}|+1.
\end{align*}
\end{linenomath}
\end{proof}

\accCRR*
\begin{proof}
We show this by induction on $i$. The induction base is settled for $i=5$ since the polynomials in \cref{eq:v5} are \crr. 

Suppose, for a contradiction, that $\N^\circ_i$ and $\D_i$, share a common real root. Then, we are either in case (i) or (ii) of \cref{lem:nofactors}.
The fact that $\D_\textsc{f}=\D_\textsc{l}=1$ excludes case (i).

Thus, we are in case (ii). To arrive at the final contradiction, we work a little harder.
Since $\D_\textsc{f}=\D_\textsc{l}=1$, there exists $z$ which is a zero of both $\{\N^\circ_\textsc{l}\D_\textsc{m}-\N^\circ_\textsc{m},\{\N^\circ_\textsc{f}\D_\textsc{m}-\N^\circ_\textsc{m}\}$. 
This implies that either $\D_\textsc{m}(z)=0$ or $\N^\circ_\textsc{l}(z)=\N^\circ_\textsc{f}(z)$.
In the first case, $\D_\textsc{m}(z)=0$, it follows that $\N^\circ_\textsc{m}[z]=0$. This is an immediate contradiction to the fact that $\D_\textsc{m}$ and $\N^\circ_\textsc{m}$ are \crr. 

Thus it remains to consider the latter case, namely that $\N^\circ_\textsc{l}(z)=\N^\circ_\textsc{f}(z)$.
For $\circ=y$, we immediately obtain  a contradiction since  $\Ny_\textsc{l}=a<\Sigma\A=\Ny_\textsc{f}$.

For $\circ=x$ it follows that $z=1$ since $\Nx_\textsc{l}=x$, $\Nx_\textsc{f}=1$. 
Moreover, by \cref{lem:nofactors}, $z$ is a zero of $E_i$.
Consequently, it suffices to show that $E_i[1]\neq 0$. 
In order to analyze the zeros of $E_i$, we define for $i\in \{5,\dots,n\}$
\[\alpha_i:=a +\sum_{j=5}^{i-1} (a_j+b_j).\]
Recall that  $E_5[x]=(\Sigma\A)x-a$ by \cref{eq:v5}.
Thus, by \cref{lem:sameAngle}, it holds for $i\in \{5,\dots,n-1\}$ that 
\begin{equation}\label{eq:Ei}
 E_{i+1}[x]=E_i[x]-(a_i+b_i)=E_5[x]-\sum_{j=5}^{i} (a_j+b_j)=(\Sigma\A)x-\alpha_{i+1}.
\end{equation}
Since $\alpha_i<\Sigma\A$, it follows  $E_i[1]\neq0$ that for all $i\geq 5$. Consequently, $\N^\circ_i$ and $\D_i$ are \crr.
\end{proof}

\subsection{Proof of \cref{thm:LKdoubleStacking}}\label{app:double}

\doubleStacking*

\begin{proof}
We start to consider $\dS{\ell}{k}$.
 Note that the degree of all but four vertices is exactly four; namely, the degree of $B$ and $\ell$ is $k+3$, the degree of $v$ and $C$ is $\ell+3$.  
 Thus, if both $\ell$ and $k$ are odd, then $\dS{\ell}{k}$ is Eulerian and thus not area-universal as shown  in \cite[Theorem 1]{kleist1Journal}.  Since the degree depends on the planar graph, all plane graphs in $[\dS{\ell}{k}]$ are Eulerian and not area-universal if $\ell\cdot k$ is odd.
 
Assume that $\ell\cdot k$ is even. In order to show the area-universality of $\dS{\ell}{k}$, we consider the \pOrder $(A,B,C,v,1,\dots, \ell, 1'\dots, k')$ in which $k'C$ is the unique undirected edge. 
For an algebraically independent area assignment \A, we define $a:=\A(ABv)$ and 
place $v_3$ at $(1,\Sigma \A)$ and $v_4$ at $(x_4,a)$. Observe that the vertices $1,2,\dots \ell$ have the predecessors $C$ and $v$ and are locally identical to an \graphClass graph. Consequently, by \cref{lem:accDegree} and \cref{lem:accCRR}, the coordinates of vertex $\ell$ can be expressed by \crr polynomials $\N^\circ_\ell, \D_\ell$ with the degrees
\[|\Nx_\ell|=|\Ny_\ell|+1=|\D_\ell|=\ell.\]
 Since $\D_v=1$ and $C_{1'}=E_{1'}\D_v=\tilde D_{1'} F_{1'}$ by definition, see \cref{eq:EF}, it follows that $E_{1'}=\tilde D_{1'}$ and $F_{1'}=1$. As we will see it holds that $|E_{1'}|=|\tilde D_{1'}|=\ell$; this implies that  
 \[\max\{|E_{1'}|-|\D_\ell|-|\D_B|-|F_{1'}|,0\}
 =\max\{\ell-\ell-0-0,0\}=0.\]
\begin{figure}[htb]
 \centering
\begin{subfigure}{.47\textwidth}
\centering
 \includegraphics[page=5,scale=.75]{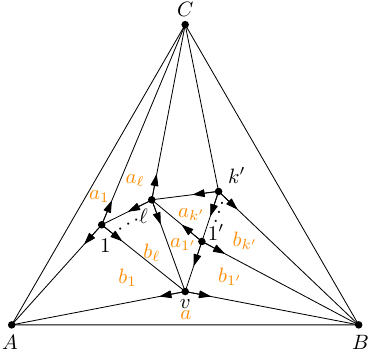}
 \caption{A \pOrder  of a   \lkGraph.}
 \label{fig:funnyDoubleA}
\end{subfigure}
\hfill
\begin{subfigure}{.48\textwidth}
\centering
 \includegraphics[page=6,scale=.75]{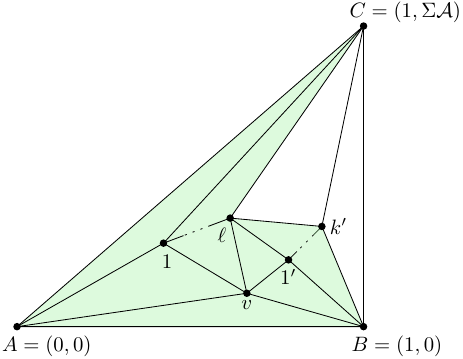}
 \caption{An \arVP.}
 \label{fig:funnyDoubleB}
\end{subfigure}
 \caption{Illustration of \cref{thm:LKdoubleStacking} and its proof.}
 \label{fig:funnyDouble}
\end{figure}
 
 Note that $d^x_B=0$, $d^y_B=-\infty$, $d^x_v=1,$ $d^y_v=0$, $d^x_\ell=0$, and $d^y_\ell=-1$.
 \cref{lem:degrees} yields the following degrees:
\[
|\Nx_{1'}|=\ell+1 \quad \text{ and }\quad |\Ny_{1'}|=|\D_{1'}|=|\tilde \D_{1'}|=\ell
\]

 
We will later show that these polynomials are \crr. Now, we proceed to compute the degrees of the vertex coordinates.
Defining $M:=\max\{|E_{1'}|,|\D_{\ell}|+|\D_B|+|F_{1'}|)\}=\max\{\ell,\ell+0+0)\}=\ell$ and by \cref{lem:degreesStacked}, 
it follows for $j>1$ 
\begin{linenomath}
\begin{align*}
  |\Nx_{j'}|&=|\D_{1'}|+(j-1)\cdot M+ d^x_{1'}=j\cdot\ell+1\\
  |\Ny_{j'}|=|\D_{j'}|&=|\D_{1'}|+(j-1)\cdot M=j\cdot\ell.
\end{align*}
\end{linenomath}
 
Assume for now, that the resulting polynomials are \crr. 
As our last face we choose the triangle $k'BC$. Then the \lastfaceFun $\f$ evaluates to 
\[\f(x):=\det(k',B,C)=1-x_{k'}=1-\frac{\Nx_{k'}}{\D_{k'}}.\]
For the last vertex $k'$, the degree of the numerator, namely $k\cdot\ell+1$, exceeds the degree of the denominator $k\cdot \ell$ and is odd since $\ell\cdot k$ is even. Consequently, $\f$ has odd max-degree and is almost surjective by \cref{lem:surjective}. Consequently,  \cref{thm:everyEmbedding} shows that every plane graph in $[\dS{\ell}{k}]$ is area-universal.\medskip

It remains to guarantee that the polynomials are \crr. 

\begin{lemma}
For all $j\geq 1$, it holds that 
 $\N^\circ_{j'}$ and $\D_{j'}$ are \crr.
\end{lemma}
 
We prove this claim by induction and start with settling the base for $j=1$ using \cref{lem:nofactors}. Suppose by contradiction that $\N^\circ_{1'}$ and $\D_{1'}$ share a common zero $z$. The fact that $\D_\textsc{l}=\D_\textsc{m}=1$ excludes case~(i). Thus, case (ii) holds and, since $\D_\textsc{l}=\D_\textsc{m}=1$, $z$ is a zero of the simplified polynomials
$(\N^\circ_\textsc{l}-\N^\circ_\textsc{m})$ and $(\N^\circ_\textsc{f}-\N^\circ_\textsc{m}\D_\textsc{f})$.
Recall that $\Nx_\textsc{l}=1$ and $\Ny_\textsc{l}=0$.
Thus for $\circ=y$, it follows that $z$ is a zero of $\Ny_\textsc{l}$ and thus also of $\Ny_\textsc{m}$. However, $\Ny_\textsc{m}=a>0$ yields a contradiction.
For $\circ=x$, $\Nx_\textsc{l}=1$ and $\Nx_\textsc{m}=x$ imply that $z=1$.
Consequently, it holds that  $\Nx_\textsc{f}[1]-\D_\textsc{f}[1]=0$. Recall that in our case $\textsc{f}=\ell$.
By \cref{obs:wiggle}, $\Nx_\textsc{f}=\Nx_\ell$ depends on $a_\ell$ while $\D_\textsc{f}$ does not. Consequently, $\Nx_\textsc{f}[1]$ and $\D_\textsc{f}[1]$ are polynomials in $\A$; due to the algebraic independence they cannot coincide.
\medskip

Now, we come to the induction step and suppose, for a contradiction, that $\N^\circ_{j'+1}$ and $\D_{j'+1}$ share a common real root $z$.
By \cref{lem:nofactors} we distinguish two cases. In all cases $z$ is zero of $E_{j'+1}$.
By \cref{lem:sameAngle}, we know that for $j\in[k-1]$ it holds that $E_{j'+1}=E_{j'}-(a_{j'}+b_{j'})\D_\ell.$
Together with $E_{1'}=a(\Nx_\ell-\D_\ell)+(1-x)\Ny_\ell$, we obtain 
\[E_{j'+1}=a\Nx_\ell+(1-x)\Ny_\ell-\left(a+\sum_{k=1}^j (a_{k'}+b_{k'})\right)\D_\ell.\]
We claim that $z$ does not depend on $a_{j'}$ and $b_{j'}$. Then, it follows from the algebraic independence,  that $z$ is a zero of both $a\Nx_\ell+(1-x)\Ny_\ell$ and $\D_\ell$.

\medskip

To prove this claim we distinguish the cases suggested by \cref{lem:nofactors}. 
Recall that the predecessor indices $\textsc{f,m,l}$ of $j'+1$ are given by $\ell,j',2$.
If case (i) of \cref{lem:nofactors} holds, then $z$ is a zero of $\D_{\ell}$ and $\D_{j'}$ since $\D_2=1$. Then clearly $z$ does not depend on $a_{j'}$ and $b_{j'}$ since $\D_\ell$ does not.

If case (ii) of \cref{lem:nofactors} holds, then $z$ is a zero of  $\N^\circ_2\D_{j'}-\N^\circ_{j'}\D_2$ and $\N^\circ_{\ell}\D_{j'}-\N^\circ_{j'}\D_{\ell}$. We distinguish two cases for $\circ\in\{x,y\}$.
For $\circ=y$, it holds that  $\Ny_2=0$ and $\D_2=1$. Thus, it follows that $\Ny_{j'}[z]=0$ and $(\Ny_{\ell}\D_{j'})[z]=0$.
Since $\Ny_{j'}$ and $\D_{j'}$ are \crr by the induction hypothesis, it holds that $\Ny_{\ell}[z]=0$. Then as a zero of $\Ny_{\ell}$, $z$ does neither depend on $a_{j'}$ nor on $b_{j'}$.

For $\circ=x$,  $\Nx_2=1$ and $\D_2=1$ imply that $\Nx_{j'}[z]=\D_{j'}[z]$ 
and $\D_{j'}[z](\Nx_{\ell}-\D_{\ell})[z]=0$. Since $\D_{j'}[z]\neq 0$, as otherwise $\Nx_{j'}$ and $\D_{j'}$ are not \crr, it holds that $\Nx_{\ell}[z]=\D_{\ell}[z]$.
Using the last fact, 
$E_{j'+1}[z]$ simplifies to $(1-z)\Ny_\ell[z]-\left(\sum_{k=1}^j (a_{k'}+b_{k'})\right)\D_\ell[z]=0.$
This implies that 
\[\Nx_{\ell}[z]
=\D_{\ell}[z]
=\frac{1}{\sum_{k=1}^j (a_{k'}+b_{k'})}(1-z)\Ny_\ell[z].\]
Since neither $\Nx_{\ell}$ nor $\D_{\ell}$ depend on $a_{k'}$ and $b_{k'}$, these three polynomial do not coincide at $z$ for small variations of $a_{k'}$. Thus for a dense set of algebraically independent area assignments, these three polynomials share no common real root. Consequently, we can assume that $z$ does not depend on $a_{j'}$ and $b_{j'}$ and thus $z$ is a zero of both $(a\Nx_\ell+(1-x)\Ny_\ell$) and $\D_\ell$.
 
However, we show that this is not the case.
\begin{clm}
 $\D_\ell$ and $(a\Nx_\ell+(1-x)\Ny_\ell)$ are \crr.
\end{clm}

Suppose $z$ is a zero of $\D_{\ell}$ and $a\Nx_\ell+(1-x)\Ny_\ell$.
Recall that by \cref{eq:Ei} and since  $\D_{i+1}=E_{i+1}D_i$ it follows for $i\in[\ell]$ 
for $i\in[\ell]$ it holds that 
\[E_{i}=x-\alpha_{i} \quad \text{ and } \quad \D_i=\prod_{j=1}^i E_j=\prod_{j=1}^{i} (x-\alpha_j).\]
Therefore, the zero set of $\D_i$ is given by  $\{\alpha_i\mid i\in [\ell]\}$.
We define 
\[G_j[x]:=a\Nx_j[x]+(1-x)\Ny_j[x]\]
and aim to show by induction on $j\in[\ell]$ that  for all $i\leq j$: $G_j[\alpha_i]\neq 0$. Note that the claim is equivalent to $G_\ell[\alpha_i]\neq 0$ for all $i\leq \ell$.
For the induction base, \cref{eq:v5} shows that $\N^x_{1}=a_1x+b_1$ and $\N^y_{1}=a_1a+b_1$. Consequently,  it holds that 
$G_1[\alpha_1]=G_1[a]=a_1a+b_1\neq 0.$
By \cref{eq:NomDen}, for $i\in[\ell-1]$, the numerator polynomials can be expressed by
$\N^\circ_{j}=\D_{j}(a_j\N^\circ_v+b_j) + \N^\circ_{j}E_{j+1}$.
This yields 
\begin{linenomath}
\begin{align*}
 G_j[\alpha_i]&=a\big(\D_j(a_j\Nx_v+b_j) + \Nx_jE_{j+1}\big)[\alpha_i]+(1-\alpha_i)\big(\D_j(a_j\Ny_v+b_j) + \Ny_j E_{j+1}\big)[\alpha_i]\\
 &=\D_{j}[\alpha_i]\cdot(aa_j+b_j(1+a-\alpha_i))+
 E_{j+1}[\alpha_i]\cdot\big(a\Nx_{j}[\alpha_i]+(1-\alpha_i)\Ny_{j}[\alpha_i]\big)
\end{align*}
\end{linenomath}

If $i\leq j$, then the first summand vanishes since $\D_{j}[\alpha_i]=0$. The second summand does not vanish by induction and since $E_{j+1}[\alpha_i]=\alpha_i-\alpha_{j+1}<0$. 
For~$i=j+1$, the second term vanishes since $E_{j+1}[\alpha_{j+1}]=0$ and the first term does not vanish since both factors do not. Consequently, it holds that $G_\ell[\alpha_i]\neq 0$. This finishes both, the proof of the claim and the theorem.
\end{proof}

%% file: 4CoordinateLemmasNotation.tex
\subsection{Analyzing the Coordinates and their Degrees}\label{sec:coordinates}
Throughout this section, let $T$ be a plane 4-connected triangulation with a \pOrder and $\A$ an algebraically independent area assignment.  
We use \cref{lem:almostReal} to obtain an  almost realizing drawing $D(x_4)$ and want to use \cref{lem:surjective} to guarantee almost surjectivity of the \lastfaceFun $\f$. Thus, we are interested in the max-degree of $\f$. 

As shown in \cref{lem:almostReal},  we can represent the coordinates $(x_i,y_i)$ of each vertex $v_i$ and the \lastfaceFun by rational functions. Specifically, we have a representation of $x_i$ and $y_i$ by polynomials $\Nx_i,\Ny_i,\D_i$ in $x_4$ such that
\[x_i=\frac{\Nx_i}{\D_i} \text{\quad and \quad} y_i=\frac{\Ny_i}{\D_i}.\]
Due to \cref{lem:surjective}, we aim for the fact that $\Nx_i$ and $\Ny_i$ are \crr with $\D_i$ and are interested in their degrees. 
As before, we denote the degree of a polynomial~$p$ by~$|p|$.
Moreover, we say that a polynomial $p(x_1,\dots,x_k)$ \emph{depends} on $x_j$ if and only if 
$p(x_1,\dots,x_j,\dots,x_k)\neq p(x_1,\dots,0,\dots,x_k)$.
\cref{eq:NomDen,eq:NomDen2} show:
\begin{obs}\label{obs:wiggle}
For $i\in\{4,\dots, n\}$, $\N^\circ_i$ depends on $a_i$ and $b_i$, while $\D_i$ does not.
\end{obs}
In order to study their degree, we define  $d_i^\circ:=|\N^\circ_i|-|\D_i|$.
\begin{lemma}\label{lem:degrees}
Let $v_i$ be a vertex with the three predecessors $p_\textsc{f},p_\textsc{m},p_\textsc{l}$ in \Po.
For the vertex coordinates of $v_i$ in $D(x_4)$, it holds that the (not necessarily \crr) polynomials $\N^\circ_i,\D_i,\tilde \D_i$, defined in \cref{eq:tilde,eq:NomDen,eq:NomDen2}, have the following degrees:
\begin{linenomath}
\begin{align*}
 |\N^\circ_i|=&|\D_\textsc{m}|+|\D_\textsc{f}|+|\D_\textsc{l}|+|F_i|\\
 &+\max\big\{d^\circ_\textsc{l},d^\circ_\textsc{f}, d^\circ_\textsc{m}+\max\{|E_i|-|\D_\textsc{f}|-|\D_\textsc{l}|-|F_i|,0\}\big\}\\
 |\D_i|=&|\D_\textsc{m}|+|E_i|=|\tilde \D_{i}|+|F_i|\\
 |\tilde \D_{i}|=&|\D_\textsc{f}|+|\D_\textsc{m}|+|\D_\textsc{l}|
 \\&+\max\big\{
			   d^x_\textsc{f}+\max\{d^y_\textsc{m},d^y_\textsc{l}\},
			   d^x_\textsc{m}+\max\{d^y_\textsc{l},d^y_\textsc{f}\},
			   d^x_\textsc{l}+\max\{d^y_\textsc{f},d^y_\textsc{m}\}\big\}
\end{align*} 
\end{linenomath}
\end{lemma}
\begin{proof}
 We need to determine the degrees of the polynomials in \cref{eq:tilde,eq:NomDen,eq:NomDen2}.
 Here we use the fact that for all polynomials $p,q$ which are not the $0$-polynomial it holds 
$|p\cdot q|=|p|+|q|.$
With the convention that $|0|=-\infty$, the above identity also holds for the $0$-polynomial. Moreover, unless $|p|=|q|$ and the leading coefficients are canceling, it holds that
$|p+q|=\max\{|p|,|q|\}.$
By algebraic independence, cancellation of leading coefficients does not occur.
%
%
%
%
\end{proof}

However, in order to apply \cref{lem:surjective}, we also need that $\Nx_i$ and $\Ny_i$ are \crr with~$\D_i$. 
Therefore, we are interested in sufficient conditions.
\begin{restatable}{lemma}{nofactors}
\label{lem:nofactors}
 Let $\circ\in\{x,y\}$. Suppose $\N^\circ_j$ and $\D_j$ are \crr for all $j<i$.
 Then, $\N^\circ_i$ and $\D_i$ have a common zero $z$ if and only if the following properties hold 
 \begin{compactitem}
 \item[-] $z$ is a zero of $E_i$
 \item[-] $z$ is independent of $a_i$ and $b_i$
 \end{compactitem}
 and additionally 
 \begin{compactitem}
  \item [(i)] $z$ is a zero of at least two of  $\{\D_\textsc{f},\D_\textsc{m},\D_\textsc{l}\}$ or
  \item [(ii)] $z$ is a zero of both of  $\{\N^\circ_\textsc{l}\D_\textsc{m}-\N^\circ_\textsc{m}\D_\textsc{l},\N^\circ_\textsc{f}\D_\textsc{m}-\N^\circ_\textsc{m}\D_\textsc{f}\}$.
 \end{compactitem}
\end{restatable}
Due to their technicality, we have moved the proofs of \cref{lem:nofactors} and of the two following lemmas to \cref{sec:appendixUsefulLemmas}.
Now, we study a more specific situation which occurs for \graphClass and double stacking graphs.
\paragraph{Stacking on same angle}
Recall that, by \cref{lem:PolyMethodInsertion}, vertex $v_i$ is inserted in a 4-face of $T_{i-1}$.
In this section, we analyze the situation that in the \pOrder several vertices are repeatedly inserted in the same angle.  In particular, we say vertices $v_{i+1}$ and $v_{i+2}$ are \emph{stacked on the same angle} if their first and last predecessors are identical and the middle predecessor of $v_{i+2}$ is $v_{i+1}$. Specifically, $v_{i+1}$ and  $v_{i+2}$ have predecessors $(p_\textsc{f},v_i,p_\textsc{l})$ and $(p_\textsc{f},v_{i+1},p_\textsc{l})$, respectively. \cref{fig:stacking} illustrates two vertices which are stacked on the same angle. 

\begin{figure}[htb]
 \centering
 \includegraphics[page=3,scale=.75]{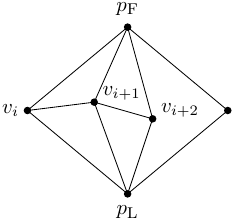}
 \caption{Vertices $v_{i+1}$ and $v_{i+2}$ are stacked on the same angle.}
 \label{fig:stacking}
\end{figure}

\begin{restatable}{lemma}{sameAngle}\label{lem:sameAngle} 
 If $v_{i+1}$ and $v_{i+2}$ are stacked on the same angle in the \pOrder and $\N^\circ_{i+1}$ and $\D_{i+1}$ 
 are \crr, then it holds that
 \begin{linenomath}
 \begin{align*}
  E_{i+2}&=E_{i+1}-(a_{i+1}+b_{i+1})\D_\textsc{f}\D_\textsc{l}F_{i+1}\\
  F_{i+2}&=F_{i+1}.
 \end{align*}
 \end{linenomath}
\end{restatable}
For the proof of \cref{lem:sameAngle}, we refer to \cref{sec:appendixUsefulLemmas}.
For the degrees, we obtain the following expressions.
 \begin{restatable}{lemma}{degreesStacked}\label{lem:degreesStacked} 
 If $v_{i+1}$ and $v_{i+2}$ are stacked on the same angle in the \pOrder, and if $\N^\circ_{i+1}$,$\D_{i+1}$, and  $\N^\circ_{i+2}$, $\D_{i+2}$ are \crr. Then, for $M:=\max\{|E_{i+1}|,|\D_\textsc{f}|+|\D_\textsc{l}|+|F_{i+1}|)\}$ and $\circ\in\{x,y\}$ it holds that
 \begin{linenomath}
 \begin{align*}
 |\N^\circ_{i+2}|&=|\D_{i+1}|+M+ d^\circ_{i+1}\\
 |\D_{i+2}|&=|\D_{i+1}|+M.
 \end{align*}
 \end{linenomath}
 In particular, it holds that $d^\circ_{i+2}=d^\circ_{i+1}$.

\end{restatable} 
For the proof of \cref{lem:degreesStacked}, we refer to \cref{sec:appendixUsefulLemmas}.

\input{4appendix}

%% file: 4appendix.tex
\subsection{Proofs of the Degree-Lemmas}\label{sec:appendixUsefulLemmas}
In this section, we present the pending proofs of the previous section.
\nofactors*
\begin{proof} 
The proof consists of two directions. For both recall the formulas for $\N^\circ_i$ and $\D_i$ given in \cref{eq:NomDen,eq:NomDen2}. Suppose that  $\N^\circ_i$ and $\D_i$ have a common zero $z$. We think of $z$ as an algebraic function of ($a,a_5,b_5,\dots,a_i,b_i$). By \cref{obs:wiggle},  the variables $a_i$ and $b_i$ do not occur in $D_i$. Consequently,  $z$ does not depend on $a_i$ and $b_i$, and is 
 thus algebraically independent of $a_i$ and $b_i$.
Since $\D_i=\D_\textsc{m} E_i$, by \cref{eq:NomDen2},  $z$ is a zero of at least one of $\D_\textsc{m}$ or $E_i$.  We distinguish three cases.  \smallskip
 
Case 1: If $z$ is a zero of both, $\D_\textsc{m}$ and $E_i$, then \cref{eq:NomDen} simplifies to
\[\N^\circ_i[z]=-(a_i+b_i)(\N^\circ_\textsc{m}\D_\textsc{f}\D_\textsc{l}F_i)[z]=0.\]
By the assumption of being \crr, $z$ is neither a zero of $\N^\circ_\textsc{m}$ nor $F_i$. Hence, $z$ is a zero of $\D_\textsc{f}\D_\textsc{l}$. In conclusion, $z$ is a zero of $E_i$, $\D_\textsc{m}$ as well as $\D_\textsc{f}$ or $\D_\textsc{l}$ (or both). In other words, condition (i) is fulfilled.\smallskip
 
Case 2: If $z$ is a zero of $\D_\textsc{m}$ and not of $E_i$, then \cref{eq:NomDen} reads as 
\[\N^\circ_i[z]=(\N^\circ_\textsc{m}(E_i-(a_i+b_i))\D_\textsc{f}\D_\textsc{l}F_i)[z]=0.\]
Since $z,a_{i},b_{i}$ are algebraically independent, $\N^\circ_i[z]$ vanishes on each summand. However, $z$ is not a zero of $E_i$ by assumption of Case 2 and   $z$ is not a zero of $\N^\circ_\textsc{m}$ since $\D_\textsc{m}$ and $\N^\circ_\textsc{m}$ are \crr. Thus, we arrive at a  contradiction and this case does not occur.\smallskip
 
Case 3: If $z$ is a zero of $E_i$ and not of $\D_\textsc{m}$, $z$ is not a zero of $F_i$; since $E_i$ and $F_i$ are \crr. Consequently, \cref{eq:NomDen} implies
$(\D_\textsc{m}(a_i\N^\circ_\textsc{l}\D_\textsc{f}+b_i \N^\circ_\textsc{f}\D_\textsc{l}) - (a_i+b_i)\N^\circ_\textsc{m}\D_\textsc{f}\D_\textsc{l})[z]=0$. 
Reordering for $a_i$ and $b_i$ results in 
\[\big(a_i\D_\textsc{f}(\D_\textsc{m}\N^\circ_\textsc{l}-\N^\circ_\textsc{m}\D_\textsc{l}) + b_i\D_\textsc{l} (\D_\textsc{m} \N^\circ_\textsc{f} - \N^\circ_\textsc{m}\D_\textsc{f})\big)[z]=0.\]
As argued above, $z$, $a_{i}$, and $b_{i}$ are algebraically independent, and thus $z$ is a zero of both summands.
If $z$ is a zero of $\D_\textsc{f}$, then it is also a zero of $\D_\textsc{l}\D_\textsc{m} \N^\circ_\textsc{f}$. However, by assumption, it is not a zero of $\D_\textsc{m} \N^\circ_\textsc{f}$, and thus it is a zero of $\D_\textsc{l}$. Likewise, if $z$ is a zero of $\D_\textsc{l}$ then it follows that $z$ is also a zero of $\D_\textsc{f}$. Hence, (i) is satisfied.

Thus, in the following we may assume that $z$ is not a zero of $\D_\textsc{f}\D_\textsc{l}$, but of both polynomials  $(\N^\circ_\textsc{l}\D_\textsc{m}-\N^\circ_\textsc{m}\D_\textsc{l})$ and $( \N^\circ_\textsc{f}\D_\textsc{m} - \N^\circ_\textsc{m}\D_\textsc{f})$. This is condition (ii).

\smallskip

It remains to show the reverse direction. Since $z$ is a zero of $E_i$, it follows that $z$ is a zero of $\D_i$. By construction of our cases, $z$ was a zero of $\N^\circ_i$. Alternatively, it is easy to check that $z$ is also a zero of $\N^\circ_i$ as given in \cref{eq:NomDen} in all cases.
%
%
%
%
%
%
 Consequently, $z$ is a zero of both $\N^\circ_i$ and $\D_i$.
\end{proof}


\sameAngle*
\begin{proof} 
Note that the irreducibility of $E_{i+1}$ and $F_{i+1}$ directly implies the irreducibility of $E_{i+2}$ and $F_{i+2}$: Since $F_{i+2}=F_{i+1}$, every zero of $F_{i+2}$ is a zero of $F_{i+1}$. Therefore, $z$ is a zero of $E_{i+2}$ and $F_{i+2}$ if and only if $z$ is a zero of $E_{i+1}$ and $F_{i+1}$.
It remains to show that 
\[E_{i+2}\D_{i+1}=F_{i+2}\tilde \D_{i+2}.\]
To do so, we will show that $\tilde \D_{i+1}$ is a factor of $\tilde \D_{i+2}$.
By \cref{eq:tilde}, we obtain the following formula for $\tilde \D_{i+2}$:
\begin{linenomath}
 \begin{align*}
  \tilde \D_{i+2}
  =& \Nx_{i+1}(\Ny_\textsc{l}\D_\textsc{f} -\Ny_\textsc{f}\D_\textsc{l})+ \Ny_{i+1}(\Nx_\textsc{f}\D_\textsc{l}-\Nx_\textsc{l}\D_\textsc{f}) +\D_{i+1}\textcolor{teal}{(\Nx_\textsc{l}\Ny_\textsc{f}-\Nx_\textsc{f}\Ny_\textsc{l})}
 \end{align*}
 \end{linenomath}
 The last summand is clearly divisible by $\D_{i+1}$ and, hence by $\tilde \D_{i+1}$. Therefore we focus on the first two summands, in which we replace $\N^\circ_{i+1}$ with the help of \cref{eq:NomDen} by 
 \begin{linenomath}
 \begin{align*}
  \N^\circ_{i+1}&=F_{i+1}\D_i(a_{i+1}\N^\circ_\textsc{l}\D_\textsc{f}+b_{i+1}\N^\circ_\textsc{f}\D_\textsc{l}) + \N^\circ_iE_{i+2}.
 \end{align*}
 \end{linenomath}
 Replacing $\N^\circ_{i+1}$ by the above expressions and using \cref{eq:tilde} for $\tilde \D_{i+1}$, we  obtain:
 \begin{linenomath*}\begin{align*}
  \Nx_{i+1}&(\Ny_\textsc{l}\D_\textsc{f} -\Ny_\textsc{f}\D_\textsc{l})+ \Ny_{i+1}(\Nx_\textsc{f}\D_\textsc{l}-\Nx_\textsc{l}\D_\textsc{f}) \\
  =&(F_{i+1}\D_i[\textcolor{violet}{a_{i+1}\Nx_\textsc{l}\D_\textsc{f}}+\textcolor{magenta}{b_{i+1} \Nx_\textsc{f}\D_\textsc{l}}]
  +\Nx_iE_{i+2})(\Ny_\textsc{l}\D_\textsc{f} -\Ny_\textsc{f}\D_\textsc{l})
  \\
  +&  (F_{i+1}\D_i[\textcolor{violet}{a_{i+1}\Ny_\textsc{l}\D_\textsc{f}}+\textcolor{magenta}{b_{i+1} \Ny_\textsc{f}\D_\textsc{l}}]+\Ny_iE_{i+2})(\Nx_\textsc{f}\D_\textsc{l}-\Nx_\textsc{l}\D_\textsc{f})
  \\
  =&{\D_\textsc{f}\D_\textsc{l}\D_iF_{i+1}(\Ny_\textsc{l}\Nx_\textsc{f}-\Ny_\textsc{f}\Nx_\textsc{l})(\textcolor{violet}{a_{i+1}}+\textcolor{magenta}{b_{i+1}})}\\&
 +E_{i+2}(\textcolor{orange}{\Nx_i(\Ny_\textsc{l}\D_\textsc{f} -\Ny_\textsc{f}\D_\textsc{l})+\Ny_i(\Nx_\textsc{f}\D_\textsc{l}-\Nx_\textsc{l}\D_\textsc{f})})\\
  =&{\D_\textsc{f}\D_\textsc{l}\D_iF_{i+1}(\Ny_\textsc{l}\Nx_\textsc{f}-\Ny_\textsc{f}\Nx_\textsc{l})(a_{i+1}+b_{i+1})}\\&
 +\textcolor{teal}{E_{i+2}}(\textcolor{orange}{\tilde \D_{i+1}-\D_{i}(\Nx_\textsc{l}\Ny_\textsc{f}-\Nx_\textsc{f}\Ny_\textsc{l})})
\end{align*} \end{linenomath*}
We remove the term $E_{i+2}\tilde \D_{i+1}$ which clearly divisible by $\tilde \D_{i+1}$. In the remainder we factor out $(\Nx_\textsc{f}\Ny_\textsc{l}-\Nx_\textsc{l}\Ny_\textsc{f})\D_i$ and recall the definitions of $E_{i+2}$ and $D_{i+1}$. It remains
 \begin{linenomath}\begin{align*}
  &{\D_\textsc{f}\D_\textsc{l}\textcolor{red}{\D_i}F_{i+1}\textcolor{red}{(\Nx_\textsc{f}\Ny_\textsc{l}-\Nx_\textsc{l}\Ny_\textsc{f})}(a_{i+1}+{b_{i+1}})}+E_{i+2}(\textcolor{red}{-\D_{i}(\Nx_\textsc{l}\Ny_\textsc{f}-\Nx_\textsc{f}\Ny_\textsc{l})})\\
 &=\textcolor{red}{(\Nx_\textsc{f}\Ny_\textsc{l}-\Nx_\textsc{l}\Ny_\textsc{f})\D_i}(\D_\textsc{f}\D_\textsc{l}F_{i+1}(a_{i+1}+b_{i+1})+E_{i+2})\\
 &=(\Nx_\textsc{f}\Ny_\textsc{l}-\Nx_\textsc{l}\Ny_\textsc{f})\D_iE_{i+1}
 =\textcolor{teal}{(\Nx_\textsc{f}\Ny_\textsc{l}-\Nx_\textsc{l}\Ny_\textsc{f})}\D_{i+1}
 \end{align*} \end{linenomath}
 By definition, $\D_{i+1}$ divisible by $\tilde \D_{i+1}$.
 The remainder consists of three summands, namely ${(\Nx_\textsc{l}\Ny_\textsc{f}-\Nx_\textsc{f}\Ny_\textsc{l})}+E_{i+2}+{(\Nx_\textsc{f}\Ny_\textsc{l}-\Nx_\textsc{l}\Ny_\textsc{f})}$.
 Note that the first and last summand of the remainder are canceling. Consequently, the remainder is $E_{i+2}$, i.e., 
 $\tilde \D_{i+2}=\tilde \D_{i+1}E_{i+2}.$
 This directly implies that $ F_{i+1}\tilde \D_{i+2}= F_{i+1}\tilde\D_{i+1}E_{i+2}=\D_{i+1}E_{i+2}$
 and therefore finishes the proof.
\end{proof}

\degreesStacked*
\begin{proof} 
 Recall that, by definition $d^\circ_{i+2}=|\N^\circ_{i+2}|-|\D_{i+2}|$. Consequently, if the degree of $\N^\circ_{i+2}$ and $\D_{i+2}$ are as claimed, it follows directly that  $d^\circ_{i+2}=d^\circ_{i+1}$. 
 For $\circ\in\{x,y\}$, we define $m^\circ:=|\D_\textsc{f}|+|\D_\textsc{l}|+|F_{i+1}|+\max\{d^\circ_\textsc{l},d^\circ_\textsc{f}\}$. Together with \cref{lem:degrees}, the degrees 
 can be expressed as follows:
 \begin{linenomath}\begin{align*}
 |\N^\circ_{j+1}|
 &=|\D_j|+\max\big\{m^\circ, d^\circ_j+M \big\}\\
 |\D_{j+1}|&=|\D_j|+|E_{j+1}|
 \end{align*} \end{linenomath}
 This implies the following formula for $d_{j+1}^\circ$:
 \begin{linenomath}\begin{align*}
  d_{j+1}^\circ&=-|E_{j+1}|+\max\big\{m^\circ, d^\circ_i+ M \big\}
 \end{align*} \end{linenomath}
 Recall that by \cref{lem:sameAngle}, $E_{i+2}=E_{i+1}-(a_{i+1}+b_{i+1})\D_\textsc{f}\D_\textsc{l}F_j$ and $F_{i+2}=F_{i+1}$. 
 Consequently, it holds that 
 $|E_{i+2}|
 =M$
 and with the above formula it follows that $|\D_{i+2}|=|\D_{i+1}|+|E_{i+2}|=|\D_{i+1}|+M$, as claimed.
 For the numerator, we replace $d_{i+1}^\circ$ in $\N^\circ_{i+2}$:
   \begin{linenomath}\begin{align*}
   |\N^\circ_{i+2}| &=|\D_{i+1}|+\max\big\{m^\circ, d^\circ_{i+1}+M\big\}\\
   &=|\D_{i+1}|+\max\bigg\{m^\circ,  M-|E|+\max\big\{m^\circ, d^\circ_i+ M \big\}\bigg\}\\
   &=|\D_{i+1}|+ M-|E|+\max\big\{m^\circ, d^\circ_i+ M \big\}\\
   &=|\D_{i+1}|+ M+d_{i+1}^\circ
  \end{align*} \end{linenomath}
 By definition of $M$, $M-|E|$ is non-negative and hence,  the outer-maximum in line 2 is attained for the second term. The last term in the third line is exactly $d_{i+1}^\circ$. Hence the numerator degree $|\N^\circ_{i+2}|$ is of the claimed form. 
\end{proof}